%% file: shortcut.tex
\documentclass[12pt]{article}

\input{prefix}

\begin{document}

\title{Jaywalking your Dog --\\
   Computing the \Frechet Distance with Shortcuts%
   \thanks{%
      A preliminary version of this paper appeared in {\em \si{Proc.}
         23rd ACM-SIAM \si{Sympos.} Discrete Algorithms}, pages
      318--337, 2012.  The latest full version of this paper is
      available online \cite{dh-jydfd-11}. %
   }}

\author{Anne Driemel%
   \thanks{Department of Information and Computing Sciences; Utrecht
      University; The Netherlands;
      \texttt{an{}n{}e}\hspace{0cm}\texttt{\atgen{}\si{cs.uu.nl}}.
      This work has been supported by the Netherlands
      \si{Organisation} for Scientific Research (\si{NWO}) under
      \si{RIMGA} (Realistic Input Models for Geographic
      Applications).} %
   \and%
   Sariel Har-Peled%
   \SarielThanks{Work on this paper was partially supported by NSF AF
      award CCF-0915984.}  }

\date{\today}

\maketitle

\setfnsymbol{stars}

\begin{abstract}
    The similarity of two polygonal curves can be measured using the
    \Frechet distance. We introduce the notion of a more robust
    \Frechet distance, where one is allowed to shortcut between
    vertices of one of the curves. This is a natural approach for
    handling noise, in particular batched outliers. We compute a
    $(3+\eps)$-approximation to the minimum \Frechet distance over all
    possible such shortcuts, in near linear time, if the curve is
    $c$-packed and the number of shortcuts is either small or
    unbounded.
    
    To facilitate the new algorithm we develop several new tools:
    \begin{compactenum}[\qquad(A)] \item A data structure for
        preprocessing a curve (not necessarily $c$-packed) that
        supports $(1+\eps)$-approximate \Frechet distance queries
        between a subcurve (of the original curve) and a line segment.
        
	\item A near linear time algorithm that computes a permutation
        of the vertices of a curve, such that any prefix of
        $\SNumVertices{k}$ vertices of this permutation, form an
        optimal approximation (up to a constant factor) to the
        original curve compared to any polygonal curve with $k$
        vertices, for any $k > 0$.
        
	\item A data structure for preprocessing a curve that supports
        approximate \Frechet distance queries between a subcurve and
        query polygonal curve. The query time depends quadratically on
        the complexity of the query curve, and only (roughly)
        logarithmically on the complexity of the original curve.
    \end{compactenum}
    To our knowledge, these are the first data structures to support
    these kind of queries efficiently.
\end{abstract}

\section{Introduction}
\seclab{introduction}

Comparing the shapes of polygonal curves -- or sequenced data in
general -- is a challenging task that arises in many different
contexts. The \Frechet{} distance and its variants (e.g., dynamic
time-warping \cite{kp-sudtw-99}) have been used as similarity measures
in various applications such as matching of time series in databases
\cite{kks-osmut-05}, comparing melodies in music information retrieval
\cite{sgh-csi-08}, matching coastlines over time
\cite{mdbh-cmpdfd-06}, as well as in map-matching of vehicle tracking
data \cite{bpsw-mmvtd-05, wsp-anmms-06}, and moving objects analysis
\cite{bbg-dsfm-08, bbgll-dcpcs-08}.  Informally, the \Frechet distance
between two curves is defined as the maximum distance a point on the
first curve has to travel as this curve is being continuously deformed
into the second curve.  Another common description uses the following
``leash'' metaphor: Imagine traversing the two curves simultaneously
and at each point in time the two positions are connected by a leash
of a fixed length. During the traversal you can vary the speeds on
both curves independently, but not walk backwards. The \Frechet
distance corresponds to the minimum length of a leash that permits
such a traversal.

\parpic[r]{\IncGraphPage[width=0.3\linewidth]{figs}{\si{hiking}}{1}}

The \Frechet distance captures similarity under small non-affine
distortions and for some of its variants also spatio-temporal
similarity \cite{mssz-fdsl-2011}.  However, it is very sensitive to
local noise, which is frequent in real data.  Unlike similarity
measures such as the root-mean-square deviation (RMSD), which averages
over a set of similarity values, and dynamic time warping, which
minimizes the sum of distances along the curves, the \Frechet distance
is a so-called \emph{bottleneck measure} and can therefore be affected
to an extent which is generally unrelated to the relative amount of
noise across the curves.  In practice, curves might be generated by
physical tracking devices, such as GPS, which is known to be
inaccurate when the connection to the satellites is temporarily
disturbed due to atmospheric conditions or reflections of the
positioning signal on high buildings. Such inaccurate data points are
commonly referred to as ``outliers''.  Note that outliers come in
batches if they are due to such a temporary external condition.
Similarly, in computer vision applications, the silhouette of an
object could be partially occluded, and in sound recordings, outliers
may be introduced due to background sounds or breathing.  Detecting
outliers in time series has been studied extensively in the literature
\cite{mmy-rs-06}.  One may also be interested in outliers as a
deviation from a certain expected behavior or because they carry some
meaning.  It could be, for instance, that trajectories of two hikers
deviate locally, because one hiker chose to take a detour to a
panoramic view point, see the example in the figure above. Outlier
detection is inherently non-trivial if not much is known about the
underlying probability distributions and the data is sparse
\cite{ay-eeaod-05}.  We circumvent this problem in the computation of
the \Frechet distance by minimizing over all possibilities for
outlier-removal.  In a sense, our approach is similar to computing a
certain notion of partial similarity.  Unlike other partial distance
measures, the distance measure we propose is parameter-free. For
comparison, in the \emph{partial \Frechet distance}, as it was studied
by Buchin \etal \cite{bbw-eapcm-09}, one is interested in maximizing
the portions of the curves which can be matched within a certain
\Frechet distance (the parameter). In this case, the dissimilar
portions of the curves are ignored. In our case, they are replaced by
\emph{shortcuts}, which have to be matched under the \Frechet
distance.

\paragraph{The task at hand.} We are given two polygonal curves $\cX$
and $\cY$ in $\Re^d$, which we perceive as a sequence of linearly
interpolated measurement points. We believe that $\cY$ is similar to
$\cX$ but it might contain considerable noise that is occluding this
similarity. That is, it might contain erroneous measurement points
(outliers), which need to be ignored when assessing the similarity.
We would like to apply a few edit operations to $\cY$ so that it
becomes as similar to $\cX$ as possible. In the process hopefully
removing the noise in $\cY$ and judging how similar it really is to
$\cX$.  To this end, we -- conceptually -- remove subsequences of
measurement points, which we suspect to be outliers, and minimize over
all possibilities for such a removal. This is formalized in the
shortcut \Frechet distance.

\paragraph{Shortcut \Frechet distance.}
A shortcut replaces a subcurve between two vertices by a straight
segment that connects these vertices.  The part being shortcut is not
ignored, but rather the new curve with the shortcuts has to be matched
entirely to the other curve under the \Frechet distance.  As a
concrete example, consider the figure below.  The \Frechet distance
between $\cX$ and $\cY$ is quite large, but after we shortcut the
outlier ``bump'' in $\cY$, the resulting new curve $\cZ$ has a
considerably smaller \Frechet distance to $\cX$.  We are interested in
computing the minimum such distance allowing an unbounded number of
shortcuts.

Naturally, there are many other possibilities to tackle the task at
hand, for example:
\begin{compactenum}[(i)]
    \item bounding the number of shortcuts by a parameter $k$,
    \item allowing shortcuts on both curves,
    \item allowing only shortcuts between vertices that are close-by
    along the curve,
    \item ignoring the part being shortcut and maximizing the length
    of the remaining portions,
    \item allowing shortcuts to start and end anywhere along the
    curve,
    \item allowing curved shortcuts, etc.
\end{compactenum}

%

\parpic[r]{%
   \begin{minipage}{0.25\linewidth}%
       \includegraphics[page=2,width=.9\linewidth]{figs/hiking}
       \figlab{shortcut:example}%
   \end{minipage}}

If one is interested in (iii) then the problem turns into a
map-matching problem, where the start and end points are fixed and the
graph is formed by the curve and its eligible shortcuts. For this
problem, results can be found in the literature \cite{cdgnw-amm-11,
   aerw-mpm-03}. A recent result by Har-Peled and Raichel
\cite{hr-fdre-11} is applicable to the variant where one allows such
shortcuts on both curves, i.e. (ii)+(iii).  The version in (iv) has
been studied under the name of \emph{partial \Frechet distance}
\cite{bbw-eapcm-09}.

In this paper, we concentrate on the \asymmetric{} \vrestricted{}
shortcut \Frechet distance (see \secref{shortcut:frechet} for the
exact definition) because computing it efficiently seems like a first
step in understanding how to solve some of the more difficult
variants, e.g.,~(v).  Surprisingly, computing this simpler version of
the shortcut \Frechet distance is already quite challenging,
especially if one is interested in an efficient algorithm.  A more
recent result by Buchin \etal~\cite{bds-jnp-13, d-raapg-13} shows that
computing the shortcut \Frechet distance exactly is weakly \NPHard for
variant (v), where we allow shortcuts to start and end anywhere along
the curve.  Furthermore, our algorithms can be extended to
variant~(i), i.e., where at most $k$ shortcuts are allowed, see
\remrefpage{k:shortcut}.

Note that allowing shortcuts on both curves does not always yield a
meaningful measure, {especially if shortcuts on both curves may be
   matched to each other.  In particular, if one of the two curves is
   more accurately sampled and can act as a model curve, allowing
   shortcuts on only one of the two curves seems reasonable. }

\paragraph{Input model.} 
A curve $\cY$ is \emphi{$c$-packed} if the total length of $\cY$
inside any ball is bounded by $c$ times the radius of the ball.
Intuitively, $c$-packed curves behave reasonably in any resolution.
The boundary of convex polygons, algebraic curves of bounded maximum
degree, the boundary of \AlphaBetaCovered shapes \cite{e-cuabc-05},
and the boundary of $\gamma$-fat shapes \cite{d-ibucf-08} are all
$c$-packed (under the standard assumption that they have bounded
complexity).  Interestingly, the class of $c$-packed curves is closed
under simplification, see \cite{dhw-afdrc-12}. This makes them
attractive for efficient algorithmic manipulation.

Another input model which is commonly used is called low density
\cite{bksv-rimga-02}.  We call a set of line segments
$\phi$-\emphi{dense}, if for any ball the number of line segments that
intersect this ball and which are longer than the radius of the ball
is bounded by $\phi$. It is easy to see by a simple packing argument
that $c$-packed curves are $O(c)$-dense.

\paragraph{Informal restatement of the problem.} 

In the parametric space of the two input curves, we are given a
terrain defined over a grid partitioning $[0,1]^2$, where the height
at each point is defined as the distance between the two associated
points on the two curves. The grid is induced by the vertices of the
two curves. As in the regular \Frechet distance, we are interested in
finding a path between $(0,0)$ and $(1,1)$ on the terrain, such that
the maximum height on the path does not exceed some $\delta$ (the
minimum such $\delta$ is the desired distance). This might not be
possible as there might be ``mountain chains'' blocking the way.  To
overcome this, we are allowed to introduce \tunnels that go through
such obstacles. Each of these \tunnels connect two points that lie on
the horizontal lines of the grid, as these correspond to the vertices
of one curve. Naturally, we require that the starting and ending
points of such a \tunnel have height at most $\delta$ (the current
distance threshold being considered), and furthermore, the price of
such a \tunnel (i.e., the \Frechet distance between the corresponding
shortcut and subcurve) is smaller than~$\delta$. Once we introduce
these \tunnels, we need to compute a \emph{monotone} path from $(0,0)$
to $(1,1)$ in the grid which uses \tunnels.  Finally, we need to
search for the minimum $\delta$ for which there is a feasible
solution.

\paragraph{Challenge and ideas.}
Let $n$ be the total number of vertices of the input curves. A priori
there are potentially $O(n^2)$ horizontal edges of the grid that might
contain endpoints of a \tunnel, and as such, there are potentially
$O(n^4)$ different families of \tunnels that the algorithm might have
to consider.  A careful analysis of the structure of these families
shows that, in general, it is sufficient to consider one (canonical)
\tunnel per family.  Using $c$-packedness and simplification, we can
reduce the number of relevant grid edges to near linear.  This in turn
reduces the number of potential \tunnels that need to be inspected to
$O(n^2)$.  This is still insufficient to get a near linear time
algorithm. Surprisingly, we prove that if we are interested only in a
constant factor approximation, for every horizontal edge of the grid
we need to inspect only a constant number of \tunnels.  Thus, we
reduce the number of \tunnels that the algorithm needs to inspect to
near linear. And yet we are not done, as naively computing the price
of a \tunnel requires time near linear in the size of the associated
subcurve.  To overcome this, we develop a new data structure, so that
after preprocessing we can compute the price of a \tunnel in
polylogarithmic time per \tunnel. Now, carefully putting all these
insights together, we get a near linear time algorithm for the
approximate decision version of the problem.

However, to compute the minimum $\delta$, for which the decision
version returns true -- which is the shortcut \Frechet distance -- we
need to search over the critical values of $\delta$. To this end, we
investigate and characterize the critical values introduced by the
shortcut version of the problem.  Using the decision procedure, we
perform a binary search of several stages over these values, in the
spirit of \cite{dhw-afdrc-12}, to get the required approximation.

\subsection*{Our results}

\begin{compactenum}[(A)]
    \item \textbf{Computing the shortcut \Frechet distance.} %
    For a prescribed parameter $\eps > 0$, we present an algorithm for
    computing a $(3+\eps)$-approximation to the \asymmetric{}
    \vrestricted{} shortcut \Frechet distance between two given
    $c$-packed polygonal curves of total complexity $n$, see
    \defref{shortcut:f:r} for the formal definition of the distance
    being approximated.
    
    If we allow an unbounded number of shortcuts the running time of
    the new algorithm is $\RTUnbounded$, see \thmref{main} for the
    exact result. A variant of this algorithm can also handle the case
    where we allow only $k$ shortcuts, with running time
    $\RTBoundedK$, see \remrefpage{k:shortcut}.  In the analysis of
    these problems we use techniques developed by Driemel \etal{} in
    \cite{dhw-afdrc-12} and follow the general approach used in the
    parametric search technique of devising a decision procedure which
    is used to search over the critical events for the \Frechet
    distance.  The shortcuts introduce a new type of critical event,
    which we analyze in \secref{shortcut:price}.  The presented
    approximation algorithms can be easily modified to yield
    polynomial-time exact algorithms for the same problems (and for
    general polygonal curves). As such, the main challenge in devising
    the new algorithm was to achieve near linear time performance.
    Furthermore, the algorithm uses a new data structure (described
    next) that is interesting on its own merit.
    
    \item \textbf{\Frechet-distance queries between a segment and a
       subcurve.} %
    We present a data structure that preprocesses a given polygonal
    curve $\cZ$, such that given a query segment $\seg$, and two
    points $\pnt, \pnt'$ on $\cZ$ (and the edges containing them), it
    $(1+\eps)$-approximates the \Frechet distance between $\seg$ and
    the subcurve of $\cZ$ between $\pnt$ and $\pnt'$. Surprisingly,
    the data structure works for any polygonal curve (not necessarily
    packed or dense), requires near linear preprocessing time and
    space, and can answer such queries in polylogarithmic time
    (ignoring the dependency on $\eps$). See
    \thmref{segment:queries:f:r} for the exact result.
    
    \item \textbf{Universal vertex permutation for curve
       simplification.} %
    We show how to preprocess a polygonal curve in near-linear time
    and space, such that, given a number $k \in \Na$, one can compute
    a simplification in $O(k)$ time which has $K= \SNumVertices{k}$
    vertices (of the original curve) and is optimal up to a constant
    factor with respect to the \Frechet distance to the original
    curve, compared to any curve which uses $k$
    vertices. Surprisingly, this can be done by computing a
    permutation of the vertices of the input curve, such that this
    simplification is the subcurve defined by the first $K$ vertices
    in this permutation. Namely, we compute an ordering of the
    vertices of the curves by their \Frechet ``significance''.  See
    \thmref{f:r:simpl:result} for the exact result.
    
    \item \textbf{\Frechet-distance queries between a curve with $k$
       vertices and a subcurve.}  We use the above universal vertex
    permutation, to extend the data structure in (B) to support
    queries with polygonal curves of multiple segments (as opposed to
    single segments) and obtain a constant factor approximation with
    polylogarithmic query time, see \thmref{k:seg:query:result}. The
    query time is quadratic in the query curve complexity and
    logarithmic in the input curve complexity.
    
\end{compactenum}

\paragraph{Related work.}%
Assume we are given two polygonal curves of total complexity $n$ and
we are interested in computing the \Frechet distance between these
curves. The problem has been studied in many variations.  We only
discuss the results which we deem most relevant and refer the reader
to \cite{bbmm-fswd-12} for additional references.

Driemel \etal presented a near linear time $(1+\eps)$-approximation
algorithm for the \Frechet distance assuming the curves are well
behaved \cite{dhw-afdrc-12}; that is, $c$-packed.  In general,
computing the \Frechet distance exactly takes roughly quadratic time.
After publication in the seminal paper by Alt and Godau
\cite{ag-cfdbt-95}, their $O(n^2 \log n)$-time algorithm remained the
state of the art for more than a decade.  This lead Alt to conjecture
that the problem of deciding whether the \Frechet distance between two
curves is smaller or equal a given value is 3SUM-hard.  However,
recently, there has been some progress in improving upon the quadratic
running time of the decision algorithm.  First, Agarwal \etal
presented a subquadratic time algorithm for a specific variant of the
\Frechet distance \cite{aaks-dfst-13}.  Buchin \etal build upon their
work and give an algorithm for the original \Frechet distance
\cite{bbmm-fswd-12}.  Their algorithm is randomized and takes $o\pth{
   n^2 \log n}$ expected time overall to compute the \Frechet
distance. The decision algorithm they present is deterministic and
takes subquadratic time.  The only lower bound known for the decision
problem is $\Omega(n \log n)$ and was given by Buchin \etal
\cite{bbkrw-wtd-07}.  A randomized algorithm simpler than the one by
Alt and Godau, which has the same running time, but avoids parametric
search, was recently presented by Har-Peled and Raichel
\cite{hr-fdre-13}.

Buchin \etal \cite{bbw-eapcm-09} showed how to compute the
\emph{partial \Frechet distance} under the $L_1$ and $L_\infty$
metric.  Here, one fixes a threshold $\delta$, and computes the
maximal length of subcurves of the input curves that match under
\Frechet distance $\delta$.  The running time of their algorithm is
roughly $O\pth{n^3 \log n}$.  For the problem of counting the number
of subcurves that are within a certain \Frechet distance, a recent
result by \si{de Berg} \etal provides a data structure to answer such
queries up to a constant approximation factor \cite{bcg-ffq-11}.  To
the best of our knowledge the problem of computing the \Frechet
distance when one is allowed to introduce shortcuts has not been
studied before.

\paragraph{Previous work on curve simplification.} There is a large
body of literature on curve simplification. Since this is not the main
subject of the paper, we only discuss a selection of results which we
consider most relevant, since they use the \Frechet distance as a
quality measure.  Agarwal \etal \cite{ahmw-nltaa-05} give a
near-linear time approximation algorithm to compute a simplification
which is in \Frechet distance $\eps$ to the original curve and of
which the size is at most the size of the optimal simplification with
error $\eps/2$.  Abam \etal \cite{abh-sals-10} study the problem in
the streaming setting, where one wishes to maintain a simplification
of the prefix seen so far.  Their algorithm achieves an $O(1)$
competitive ratio using $O(k^2)$ additional storage and maintains a
curve with $2k$ vertices which has a smaller \Frechet distance to the
prefix than the optimal \Frechet simplification with $k$ vertices.
Bereg \etal \cite{bjwyz-spcdfd-08} give an exact $O(n\log n)$
algorithm that minimizes the number of vertices in the simplification,
but using the discrete \Frechet distance, where only distances between
the vertices of the curves are considered. Simplification under the
\Frechet distance has also been studied by Guibas \etal
\cite{ghms-apsml-93}.

\paragraph{Organization.}

In \secref{prelims} we describe some basic definitions and results. In
particular, the formal problem statement and the definition of the
\emph{\asymmetric{} \vrestricted{} shortcut \Frechet distance} between
two curves is given in \secref{shortcut:frechet}.  We also discuss
some basic tools needed for the algorithms.
In \secref{algo:unbounded:general}, we describe the approximation
algorithm for the shortcut \Frechet distance.  Here, we devise an
approximate decision procedure in \secref{algo:decider} that is used
in the main algorithm, described in \secref{algo:main}, to search over
an approximate set of candidate values.  The analysis of this
algorithm is given in \secref{analysis:unbounded}.  Since the
shortcuts introduce a new set of candidate values, we provide an
elaborate study of these new events in \secref{shortcut:price}.  The
main result for approximating the shortcut \Frechet distance is stated
in \thmref{main}.  In the remaining sections we describe the new data
structures.  In \secref{f:r:query} we describe a data structure for a
fixed curve, that answers queries for the \Frechet distance between a
subcurve and a given segment.  In \secref{f:r:simpl}, we use this data
structure to compute the universal vertex permutation. The extension
to query curves with more than two vertices are described in
\secref{k:seg:query}.  We conclude with discussion and some open
problems in \secref{conclusions}.

\section{Preliminaries}%
\seclab{prelims}

\paragraph{Notation.}
%
A \emphi{curve} $\cX$ is a continuous mapping from $[0,1]$ to $\Re^d$,
where $\cX(t)$ denotes the point on the curve parameterized by $t \in
[0,1]$.  Given two curves $\cX$ and $\cY$ that share an endpoint, let
$\cX \concatOp \cY$ denote the \emphi{concatenated} curve.  We denote
with $\SubCrv{\cX}{\x}{\x'}$ the subcurve of $\cX$ from $\cX(\x)$ to
$\cX(\x')$ and with $\SC{\cX}{\pnt}{\pnt'}$ the subcurve of $\cX$
between the two points $\pnt,\pnt' \in \cX$.  Similarly,
$\ScutCrv{\cY}{\y}{\y'}$ denotes the line segment between the points
$\cY(\y)$ to $\cY(\y')$, we call this a \emphi{shortcut} of
$\cY$. 
For a set of numbers $U$, an \emphi{atomic interval} is a maximum
interval of $\Re$ that does not contain any point of $U$ in its
interior.

\subsection{Background and standard definitions}
\seclab{background}

Some of the material covered in this section is standard, and follows
the presentation in Driemel \etal \cite{dhw-afdrc-12}.  A
\emphi{reparameterization} is a one-to-one and continuous function
$f:[0,1]\rightarrow [0,1]$. It is \emphi{orientation-preserving} if it
maps $f(0)=0$ and $f(1)=1$.  The \Frechet distance is defined only for
oriented curves, as we need to match the start and end points of the
curves. The orientation of the curves we use would be understood from
the context.

\begin{defn}%
    \deflab{width:f}%

    Let $\cX:[0,1] \to \Re^d$ and $\cY:[0,1] \to \Re^d$ be two
    polygonal curves. We define the \emphi{width} of an
    orientation-preserving reparametrization $f:[0,1] \to [0,1]$, with
    respect to $\cX$ and $\cY$, as
    \begin{align*}
        \widthX{f}{\cX,\cY} = \max_{\alpha \in [0,1]}
        \distX{\cX(f(\alpha))}{\cY(\alpha)}.
    \end{align*}
    The \emphi{\Frechet distance} between the two curves is
    \begin{align*}
        \distFr{\cX}{\cY} = \inf_{%
           \substack{f:[0,1] \rightarrow [0,1] }}~
        \widthX{f}{\cX,\cY}.
    \end{align*}
\end{defn}


\begin{defn}%
    \deflab{elevation}%
    Let $\cX:[0,1] \to \Re^d$ and $\cY:[0,1] \to \Re^d$ be two
    polygonal curves.  The square $[0,1]^2$ represents their
    \emphi{parametric space}.  For a point $\pnt = (\xPnt, \yPnt) \in
    [0,1]^2$, we define its \emphi{elevation} to be $ \frVal{\pnt} =
    \distX{\cX(\xPnt)}{\cY(\yPnt)}.$ Let $\delta>0$ be a parameter,
    the $\delta$-\emphi{free space} of $\cX$ and $\cY$ is defined as
    \begin{align*}
        \FullFDleq{\delta}(\cX, \cY) = \brc{ \pnt \in [0,1]^2 \sep{
              \frVal{\pnt} \leq \delta}}.
    \end{align*}
\end{defn}

\paragraph{Free space diagram.} We are interested only in polygonal
curves, which we assume to have uniform parameterizations.  The
parametric space can be broken into a (not necessarily uniform) grid
called the \emphi{free space diagram}, where a vertical line
corresponds to a vertex of $\cX$ and a horizontal line corresponds to
a vertex of $\cY$.

\parpic[r]{\includegraphics[scale=0.7]{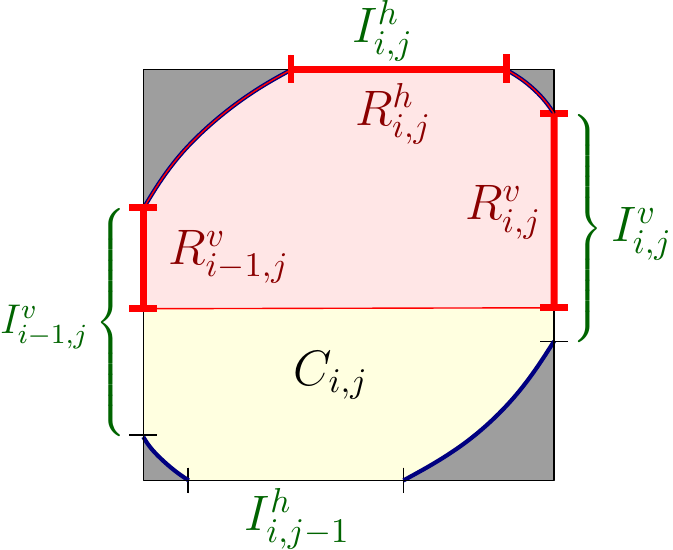}}

Every two segments of $\cX$ and $\cY$ define a \emphi{free space cell}
in this grid.  In particular, let $\CellXY{i}{j} = \CellXY{i}{j}\pth{
   \cX, \cY}$ denote the free space cell that corresponds to the
$i$\th edge of $\cX$ and the $j$\th edge of $\cY$. The cell
$\CellXY{i}{j}$ is located in the $i$\th column and $j$\th row of this
grid.

It is known that the free space, for a fixed $\delta$, inside such a
cell $\CellXY{i}{j}$ (i.e., $\FullFDleq{\delta}(\cX, \cY) \cap
\CellXY{i}{j}$) is the clipping of an affine transformation of a disk
to the cell \cite{ag-cfdbt-95}, see the figure on the right; as such,
it is convex and of constant complexity.  Let $\IHCellXY{i}{j}$ denote
the horizontal \emphi{free space interval} at the top boundary of
$\CellXY{i}{j}$, and $\IVCellXY{i}{j}$ denote the vertical free space
interval at the right boundary.

We define the \emphi{complexity} of the relevant free space, for
distance $\delta$, denoted by $\Nleq{\delta}(\cX, \cY)$, as the total
number of grid cells that have a non-empty intersection with
$\FullFDleq{\delta}(\cX,\cY)$.

\begin{observation}%
    \obslab{f:r:segments}%
    Given two segments $\pnt \pntA$ and $u v$, it holds $\distFr{\pnt
       \pntA}{u v} = \max ( \distX{u}{\pnt}$, $\distX{v}{\pntA})$.  To
    see this, consider the uniform parameterization $\pnt(t) = t \pnt
    + (1-t) \pntA$ and $u(t) = t u + (1-t) v$, for $t\in [0,1]$. It is
    easy to verify that $f(t) = \distX{\pnt(t)}{u(t)}$ is convex, and
    as such $f(t) \leq \max( f(0), f(1))$, for any $t \in [0,1]$.
\end{observation}

\paragraph{Free space events.}

To compute the \Frechet distance consider increasing $\delta$ from $0$
to $\infty$. As $\delta$ increases, structural changes happen to the
free space. We are interested in the radii (i.e., the value of
$\delta$) of these events.

\parpic[l]{\includegraphics{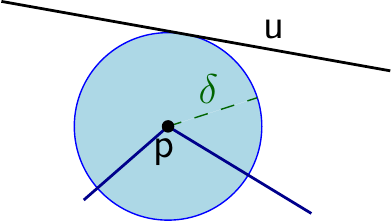}}

Consider a segment $\segA\in\cX$ and a vertex $\pnt\in\cY$, a
\emphi{vertex-edge event} corresponds to the minimum radius of a ball
centered at $\pnt$, such that $\segA$ is tangent to the ball, see the
figure on the left.  In the free space diagram, this corresponds to
the event that a free space interval consists of one point only.  The
line supporting this boundary edge corresponds to the vertex, and the
other dimension corresponds to the edge.  Naturally, the event could
happen at a vertex of $\segA$.  The second type of event, a
\emphi{monotonicity event}, corresponds to a value $\delta$ for which
a monotone subpath inside the $\delta$-free space becomes feasible.
Geometrically, this corresponds to the common distance of two vertices
on one curve to the intersection point of their bisector with a
segment on the other curve.

\subsection{The $k$-shortcut \Frechet distance}
\seclab{shortcut:frechet}

\begin{defn}%
    \deflab{shortcut:curve}%
    For a polygonal curve $\cY$, we refer to any order-preserving
    concatenation of $k+1$ non-overlapping (possibly empty) subcurves
    of $\cY$ with $k$ shortcuts connecting the endpoints of the
    subcurves in the order along the curve, as a $k$-\emphi{shortcut
       curve} of $\cY$.  Formally, for values $0 \leq \y_1 \leq \y_2
    \leq \cdots \leq \y_{2k} \leq 1$, the shortcut curve is defined as
    $\SubCrv{\cY}{0}{\y_1} + \ScutCrv{\cY}{\y_1}{\y_2} +
    \SubCrv{\cY}{\y_2}{\y_3} + \cdots +
    \ScutCrv{\cY}{\y_{2k-1}}{\y_{2k}} + \SubCrv{\cY}{\y_{2k}}{1 }$. If
    each $\cY(\y_i)$ is a vertex of $\cY,$ we refer to the shortcut
    curve as being \emphi{\vrestricted{}}, otherwise we say it is
    \emphi{\scunrestricted}.
\end{defn}

\begin{defn} %
    \deflab{shortcut:f:r}%
    Given two polygonal curves $\cX$ and $\cY$, we define their
    \emphi{continuous} $k$-shortcut \Frechet distance as the minimal
    \Frechet distance between the curve $\cX$ and any unrestricted
    $k$-shortcut curve of $\cY$.  We denote it with
    $\distSFr{k}{\cX}{\cY}$.  If we do not want to bound the number of
    shortcuts, we omit the parameter $k$ and denote it with
    $\distoSFr{\cX}{\cY}$.  The \emphi{\vrestricted{} $k$-shortcut
       \Frechet distance} is defined as above using only
    \vrestricted{} shortcut curves of $\cY$.  Furthermore, note that
    in all cases we allow only one of the input curves to be shortcut,
    namely $\cY$, thus we call the distance measure
    \emphi{\asymmetric{}}.
\end{defn}

In this paper, we study the \asymmetric{} \vrestricted{} $k$-shortcut
\Frechet distance for the case of bounded and unbounded $k$.  In the
following, we will omit the predicates \asymmetric{} and
\vrestricted{} when it is clear from the context.

\paragraph{Free space.}%

The \emphi{$k$-reachable free space} $\FDleqI{k}{\delta}\pth{\cX,\cY}$
is
\begin{align*}
    \FDleqI{k}{\delta}\pth{\cX,\cY} =%
    \brc{ \pnt = (\xPnt, \yPnt) \in [0,1]^2 \sep{
          \distSFr{k}{\SubCrv{\cX}{0}{\xPnt}}{\SubCrv{\cY}{0}{\yPnt }}
          \leq \delta}}.
\end{align*}
This is the set of points that have an $(x,y)$-monotone path from
$(0,0)$ that stays inside the free space and otherwise uses at most
$k$ \tunnels, which are defined in the next subsection.

\begin{figure*}[\si{bt}]\center
    \includegraphics{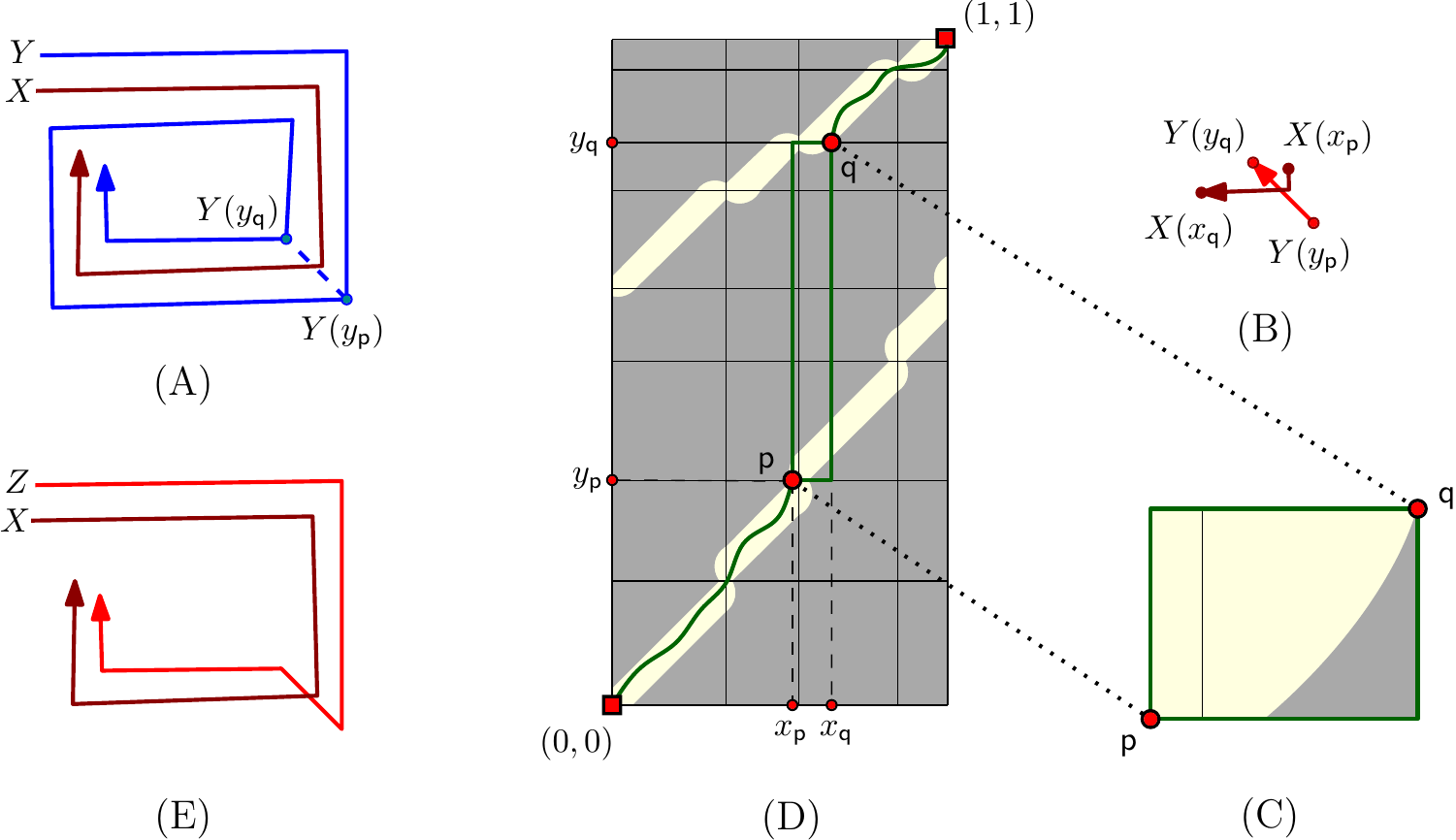} \caption{(A) Example of
       two dissimilar curves that can be made similar by shortcutting
       one of them.  (B) A \tunnel $\xtunnel{\pnt}{\pntA}$ corresponds
       to a shortcut and a subcurve matched to each other and (C)
       their free space diagram.  (D) The \tunnel connects previously
       disconnected components of the free space.  (E) The curve $\cZ$
       resulting from shortcutting $Y$. Its (regular) \Frechet
       distance from $\cX$ is dramatically reduced.  }
    \figlab{simple:example}
\end{figure*}

\subsection{\Tunnels and gates -- definitions}%
\seclab{def:tunnels}

\subsubsection{\Tunnels}%
\seclab{tunnels:real:def}

In the parametric space, a shortcut $\ScutCrv{\cY}{\yPnt}{\yPntA}$ and
the subcurve $\SubCrv{\cX}{\xPnt}{\xPntA}$, that it is being matched
to, correspond to a the rectangle with corners $\pnt$ and $\pntA$,
where $\pnt = \pth{\xPnt, \yPnt}$ and $\pntA = \pth{\xPntA, \yPntA}$.
By shortcutting the curve on the vertical axis, we are collapsing this
rectangle to a single row, see \figref{simple:example}~(C).  More
precisely, this is the free space diagram of the shortcut and the
subcurve.  We call this row a \emphi{\tunnel{}} and denote it by
$\xtunnel{\pnt}{\pntA}$.  We require $\xPnt \leq \xPntA$ and $\yPnt
\leq \yPntA$ for monotonicity.  \figref{simple:example} shows the full
example of a \tunnel. We call the \Frechet distance of the shortcut
segment to the subcurve the \emphi{price} of this \tunnel and denote
it with $\scPrice{\pnt}{\pntA} =
\distFr{\SubCrv{\cX}{\xPnt}{\xPntA}}{\ScutCrv{\cY}{\yPnt}{\yPntA}}$.
A \tunnel $\xtunnel{\pnt}{\pntA}$ is \emphi{feasible} for $\delta$ if
it holds that $\frVal{\pnt} \leq \delta$ and $\frVal{\pntA} \leq
\delta$, i.e., if $\pnt,\pntA \in \FullFDleq{\delta}({\cX},{\cY})$.
(Note that in turn the feasibility of a monotone path in the free
space of the tunnel is determined by the price of the tunnel.)  Now,
let $u=\cY(\yPnt)$ and $v=\cY(\yPntA)$ and let $\edge$ be the edge of
$\cX$ that contains $\cX(\xPnt)$ (resp., $\edge'$ the edge that
contains $\cX(\xPntA)$) for the \tunnel $\xtunnel{\pnt}{\pntA}$.  We
denote with $\xTunnels{\edge}{\edge'}{u}{v}$ the \emphi{family of
   \tunnels{}} that $\xtunnel{\pnt}{\pntA}$ belongs to.  Furthermore,
let $\xTunnelsLeq{\edge}{\edge'}{u}{v}{\delta}$ denote the subset of
these \tunnels that are feasible for $\delta$.

\begin{defn}%
    \deflab{min:radius}%
    \deflab{canonical:price}%
    The \emphi{canonical \tunnel{}} of the \tunnel family
    $\xTunnels{\edge}{\edge'}{u}{v}$, denoted by $\bCanonical{\edge}
    {\edge'}{u}{v}$, is the \tunnel that matches the shortcut $uv$ to
    the subcurve $\cX[s,t]$, such that $s$ and $t$ are the values
    realizing
    \begin{equation}
        \rMinCreate{\edge}{\edge'}{u}{v}%
        = %
        \min_{\cX(s) \in \edge,\cX(t) \in \edge', \atop s \leq
           t} \max\pth{
           \begin{array}{c}
               \distX{\cX(s)}{u},\\
               \distX{\cX(t)}{v}
           \end{array}
        }.
        \eqlab{c:tunnel:endpoints}
    \end{equation}    
    We refer to $\rMinCreate{\edge}{\edge'}{u}{v}$ as the
    \emphi{minimum radius} of this family. The canonical tunnel may
    not be uniquely defined if only one of the two values $s$ or $t$
    determines the minimum radius.  In this case, we define $s$ and
    $t$ as the values minimizing $\distX{\cX(s)}{u}$ and
    $\distX{\cX(t)}{v}$ for $\cX(s) \in \edge$ and $\cX(t)\in \edge'$,
    individually.  We call the price of the canonical \tunnel the
    \emphi{canonical price} of this \tunnel family.
\end{defn}

Clearly, one can compute the canonical \tunnel
$\xTunnels{\edge}{\edge'}{u}{v}$ in constant time. In particular, the
price of this canonical \tunnel is
\begin{equation}
    \priceX{ \bCanonical{\edge}{\edge'}{u}{v}} = 
    \distFr{\SubCrv{\cX}{s}{t}}{uv}.
    \eqlab{c:tunnel:price}
\end{equation}

We emphasize that a shortcut is always a segment connecting two
vertices of the curve $\cY$, and a \tunnel always lies in the
parametric space; that is, they exist in two completely different
domains.

\begin{observation}%
    \obslab{min:radius:eq}%
    The minimum radius of a tunnel family
    $\rMinCreate{\edge}{\edge'}{u}{v}$ corresponds to either
    \begin{inparaenum}[(i)]
        \item the distance of $u$ to its closest point on $\edge$,
        \item the distance of $v$ to its closest point on $\edge'$, or
        \item the common distance of $u$ and $v$ to the intersection
        of their bisector with the edge $\edge$ (i.e., a monotonicity
        event). Note that the event in case (iii) can only happen if
        $\edge=\edge'$.
    \end{inparaenum}
\end{observation}

\subsubsection{Gates}%
\vspace{-0.2cm}%

\parpic[r]{\includegraphics{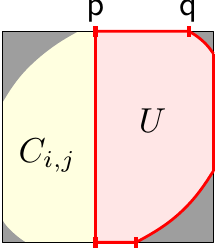}}%
\noindent Let $U$ be a subset of the parametric space that is convex
in every cell.  Let $\IHCellXY{i}{j}$ be a free space interval. We
call the left endpoints of $U \cap \IHCellXY{i}{j}$ the \emphi{left
   gate} of $U$ in the cell $\CellXY{i}{j}$, and similarly the right
endpoint is the \emphi{right gate}.  The figure to the right shows an
example of gates $\pnt$ and $\pntA$.  The \emphi{set of gates} of $U$
are the gates with respect to all cells in the free space diagram.  We
define the \emphi{canonical gate} of a vertex-edge pair as the point
in parametric space that minimizes the vertex-edge distance. Note that
canonical gates serve as endpoints of canonical tunnels that span
across columns in the free-space diagram.

\subsection{Curve simplification}

We use the following simple algorithm for the simplification of the
input curves. It is easy to verify that the curve simplified with
parameter $\sRadius$ is in \Frechet distance at most $\sRadius$ to the
original curve, see \cite{dhw-afdrc-12}.

\begin{defn}%
    \deflab{simplification}%
    Given a polygonal curve $\cXOrig$ and a parameter $\sRadius > 0$.
    First mark the initial vertex of $\cXOrig$ and set it as the
    current vertex. Now scan the polygonal curve from the current
    vertex until it reaches the first vertex that is in distance at
    least $\sRadius$ from the current vertex. Mark this vertex and set
    it as the current vertex. Repeat this until reaching the final
    vertex of the curve, and also mark it.  We refer to the resulting
    curve $\cX$ that connects only the marked vertices, in their order
    along $\cXOrig$, as a \emphi{$\sRadius$-simplification} of
    $\cXOrig$ and we denote it with $\simpX{\cXOrig,\sRadius}$.
\end{defn}

\begin{figure}\centering
    \includegraphics[width=0.3\textwidth]{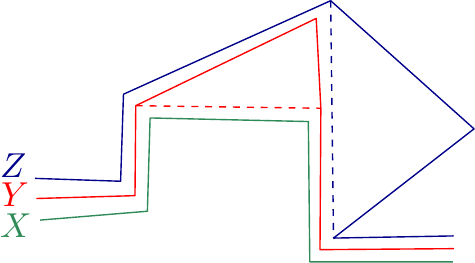}
    \caption{The \asymmetric{} $k$-shortcut \Frechet does not
          satisfy the triangle inequality. In the depicted
          counter-example it holds that $\distSFr{k}{\cX}{\cZ} >
          \distSFr{k}{\cX}{\cY} + \distSFr{k}{\cY}{\cZ}$ for any value
          of $k$ and for $k$ unbounded. This holds true in the
          \vrestricted{} and in the \unrestricted{} case.}
       \figlab{triangle:inequality}
\end{figure}

During the course of the algorithm we will simplify the input curves
in order to reduce the complexity of the free space.  The $k$-shortcut
\Frechet distance does not satisfy the triangle inequality, as can be
seen by the counter-example shown in \figref{triangle:inequality}.
Therefore, we need the next lemma to ensure that the computed distance
between the simplified curves approximates the distance between the
original curves.  The proof is straightforward and can be found in
\cite{d-raapg-13}.

\begin{lemma}[\cite{d-raapg-13}]%
    \lemlab{t:r:i:n:e:q:s:u:b}%
    Given a simplification parameter $\sRadius$ and two polygonal
    curves $\cXOrig$ and $\cYOrig$, let $\cX =
    \simpX{\cXOrig,\sRadius}$ and $\cY = \simpX{\cYOrig, \sRadius}$
    denote their $\sRadius$-simplifications, respectively.  For any $k
    \in \Na$, it holds that $\distSFr{k}{\cX}{\cY} - 2 \sRadius \leq
    \distSFr{k}{\cXOrig}{\cYOrig} \leq \distSFr{k}{\cX}{\cY} +
    2\sRadius$.  Similarly, $\distoSFr{\cX}{\cY} - 2 \sRadius \leq
    \distoSFr{\cXOrig}{\cYOrig} \leq \distoSFr{\cX}{\cY} + 2\sRadius$.
\end{lemma}

\begin{lemma}[\cite{dhw-afdrc-12}]%
    \lemlab{complexity:l:e:q}%
    For any two $c$-packed curves $\cXOrig$ and $\cYOrig$ in $\Re^d$
    of total complexity $n$, and two parameters $0 < \eps < 1$ and
    $\delta > 0$, we have that
    $\Nleq{\delta}(\simpX{\cXOrig,\sRadius},\simpX{\cYOrig, \sRadius})
    = O( cn /\eps)$, where $\sRadius=\Theta(\eps\delta)$.
\end{lemma}

\subsection{Building blocks for the algorithm}

The algorithm uses
the following two non-trivial data structures.
\begin{datastructure}%
    \dslab{d:s:magic:B}%
    Given a polygonal curve $\cZ$ with $n$ vertices in $\Re^d$, one
    can build a data structure, in $O\pth{\GridCompl^2 n \log^2 n }$
    time, using $O\pth{\GridCompl^2 n }$ space, where
    $\GridCompl=\NgridCompl$, that supports a procedure
    $\tunnelPrice{}\pth{\pnt,\pntA}$ which receives two points $\pnt$
    and $\pntA$ in the parametric space of $\cX$ and $\cY$ and returns
    a value $\phi$, such that $ \phi \leq \scPrice{\pnt}{\pntA} \leq
    (1+\eps)\phi$ in $O\pth{ \eps^{-3} \log n \log\log n}$ time.  See
    \secrefpage{f:r:query} and \thmrefpage{segment:queries:f:r}.
\end{datastructure}

\begin{datastructure}
    \dslab{d:s:magic:A}%
    For given parameters $\eps$ and $\delta$, and two $c$-packed
    curves $\cXOrig$ and $\cYOrig$ in $\Re^d$, let $\cX =
    \simpX{\cXOrig, \sRadius}$ and $\cX = \simpX{\cYOrig, \sRadius}$,
    where $\sRadius = \eps \delta$.  One can compute all the
    vertex-edge pairs of the two simplified curves $\cX$ and $\cY$ in
    distance at most $\delta$ from each other, in time $O( n \log n +
    c^2n / \eps )$. See below for details.
\end{datastructure}



We describe here how to realize \dsref{d:s:magic:A}.  Observe that
$\cX$ and $\cY$ have density $\density = O( c)$.  Now, we build the
data structure of \si{de Berg} and Streppel \cite{bs-arsbsp-06} for
the segments of $\cY$ (with $\eps=1/2$). For each vertex of $\cX$ we
compute all the segments of $\cY$ that are in distance at most
$\delta$ from it, using the data structure \cite{bs-arsbsp-06}. Each
query takes $O( \log n + k \density)$ time, where $k$ is the number of
edges reported.  \lemref{complexity:l:e:q} implies that the total sum
of the $k$'s is $O( c n /\eps)$.  We now repeat this for the other
direction. This way, one can realize \dsref{d:s:magic:A}.

\subsection{Monotonicity of the prices of tunnels}

The following two lemmas imply readily that under certain conditions
the prices of tunnels which share an endpoint are approximately
monotone with respect to the $x$-coordinate of their starting
point. We will exploit this in the approximation algorithm that
computes the reachability in the free-space diagram. We will see in
\secref{tunnel:test} that this drastically reduces the number of
tunnels that need to be inspected in order to decide if a particular
cell is reachable.

\begin{lemma}%
    \lemlab{monotone:shortcut:base}%
    Given a value $\delta > 0$ and two curves $\subcX{1}$ and
    $\subcX{2}$, such that $\subcX{2}$ is a subcurve of $\subcX{1}$,
    and given two line segments $\subsegY{1}$ and $\subsegY{2}$, such
    that $\distFr{\subcX{1}}{\subsegY{1}} \leq \delta$ and the start
    (resp., end) point of $\subcX{2}$ is in distance $\delta$ to the
    start (resp., end) point of ${\subsegY{2}}$, then
    $\distFr{\subcX{2}}{\subsegY{2}} \leq \constSC \delta$.
\end{lemma}


\noindent
\begin{minipage}{0.8\textwidth}
    \begin{proof}
        Let $\segA$ denote the subsegment of $\subsegY{1}$ that is
        matched to $\subcX{2}$ under an optimal \Frechet mapping
        between $\subcX{1}$ and $\subsegY{1}$.  We know that
        $\distFr{\subcX{2}}{\segA} \leq \delta$ by this mapping.  The
        start point of $\subsegY{2}$ is in distance $2\delta$ to the
        start point of $\segA$, since they are both in distance
        $\delta$ to the start point of $\subcX{2}$ and the same holds
        for the end points. This implies that
        $\distFr{\segA}{\subsegY{2}}\leq 2\delta$.  Now, by the
        triangle inequality, $\distFr{\subcX{2}}{\subsegY{2}} \leq
        \distFr{\subcX{2}}{\segA} + \distFr{\segA}{\subsegY{2}} \leq
        \constSC \delta$.
    \end{proof}
\end{minipage}%
\begin{minipage}{0.19\textwidth}
    \hfill {\includegraphics[width=0.9\textwidth]{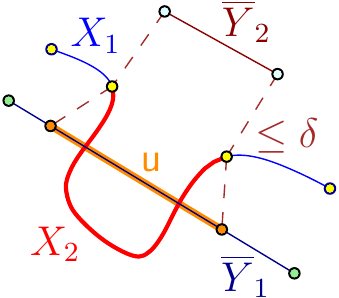}}
\end{minipage}

\begin{lemma}%
    \lemlab{monotone:shortcut}%
    Consider two polygonal curves $\cX$ and $\cY$, three points $\pnt,
    \pntA$ and $\pntB$ in their free space, and let $\delta' =
    \max(\frVal{\pnt}, \frVal{\pntA}, \frVal{\pntB})$. If $x(\pnt)
    \leq \x(\pntA) \leq \x(\pntB)$ then
    $\priceX{\xtunnel{\pntA}{\pntB}} \leq 3 \max \pth{\MakeSBig
       \delta', \priceX{\xtunnel{\pnt}{\pntB} } }$.%
    \footnote{Here, $\frVal{\pnt}$ is the elevation of $\pnt$, see
       \defref{elevation}, and $\xtunnel{\pntA}{\pntB}$ is the tunnel
       between $\pntA$ and $\pntB$, see \secref{tunnels:real:def}.}
\end{lemma}

\begin{proof}
    Let $\subcX{1}$ be the subcurve $\SubCrv{\cX}{\xPnt}{\xPntB}$, and
    let $\subcX{2} = \SubCrv{\cX}{\xPntA}{\xPntB}$.  Similarly, let
    $\subsegY{1}$ be the shortcut $\ScutCrv{\cY}{\yPnt}{\yPntB}$ and
    let $\subsegY{2} = \ScutCrv{\cY}{\yPntA}{\yPntB}$. By
    \lemref{monotone:shortcut:base} $\priceX{\xtunnel{\pntA}{\pntB}}
    \leq 3\delta'$ for $\delta' = \max \pth{
       \distFr{\subcX{1}}{\subsegY{1}}, \delta}$.
\end{proof}

We will see in \secref{second:stage} that the monotonicity property of
\lemref{monotone:shortcut} also enables a faster search over tunnel
events. The property holds even if the \tunnel{}s under consideration
are not valid. For example if $\xPnt < \xPntB$ and $\yPnt > \yPntB$
then the tunnel $\xtunnel{\pnt}{\pntB}$ is not a valid tunnel and it
cannot be used by a valid solution. Nevertheless,
$\xtunnel{\pnt}{\pntB}$ has a well defined price, and these prices
have the required monotonicity property.

\begin{lemma}
    \lemlab{monotone:matrix}%
    For a parameter $\delta \geq 0$, let $\pnt_1, \ldots, \pnt_m$ be
    $m$ points in the $\delta$-free space in increasing order by their
    $x$-coordinates, and let $\psi_i =
    \priceX{\xtunnel{\pnt_{i}}{\pnt_{m}}}$ for any $1\leq i\leq m$.
    Then, we have:
    \begin{compactenum}[\qquad (A)]
        \item If $\psi_i \geq \delta$ then for all $j> i$, we have
        $\priceX{\xtunnel{\pnt_{j}}{\pnt_{m}}} \leq 3\psi_i$.
        
        \item If $\psi_i > 3\delta$ then for all $j< i$, we have
        $\priceX{\xtunnel{\pnt_{j}}{\pnt_{m}}} \geq \psi_i/3$.
    \end{compactenum}
\end{lemma}

\begin{proof}
    To see the first part of the claim, note that by
    \lemref{monotone:shortcut}, $\priceX{\xtunnel{\pnt_{j}}{\pnt_{m}}}
    \leq 3 \max\pth{\delta,\psi_i} \leq 3 \psi_i$. As for the second
    part, we have by the same lemma that $ \delta < \psi_i/3 \leq
    \max\pth{\delta,\priceX{\xtunnel{\pnt_{j}}{\pnt_{m}}} } $, and
    thus $\psi_i/3 \leq \priceX{\xtunnel{\pnt_{j}}{\pnt_{m}}}$.
\end{proof}

\section{Approximating the shortcut \Frechet distance}%
\seclab{algo:unbounded:general}%

We describe the algorithm to approximate the \asymmetric{} \vrestricted{}
shortcut \Frechet distance between two given polygonal curves $\cX$ and $\cY$
where the number of shortcuts that can be used in a solution is unbounded.
In \secref{analysis:unbounded} we prove the correctness and analyze the
complexity of this algorithm.

%

\subsection{The \tunnelTest{} procedure}
\seclab{tunnel:test}

A key element in the decision procedure is the \tunnelTest{} procedure
depicted in \figref{tunnel:test}.  During the decision procedure, we
will repeatedly invoke the \tunnelTest procedure with a set of gates
$\Reached$, for which we already know that they are contained in the
reachable free space $\FDleqI{\infty}{\delta}\pth{\cX,\cY}$, and the
left gate associated with a horizontal free space interval of
$\FullFDleq{\delta}(\cX,\cY)$, in order to determine, if and to which
extent this interval is reachable.

Intuitively, this procedure receives as input a set of reachable
points in the parametric space and a free space interval (in the form
of the left gate) and we are asking if there exists an
\emphi{affordable} \tunnel connecting a reachable point to the
interval. Here, affordable means that its price is less than $\delta$.
More precisely, the procedure receives a set of gates $\Reached$ and a
gate $\pnt$ as input and returns the endpoint of an (approximately)
affordable tunnel that starts at a gate of $\Reached$ and ends either
at $\pnt$ or the leftmost point to the right of $\pnt$ in the same
free space interval.  If a tunnel between a gate in $\Reached$ and the
free space interval of $\pnt$ exists, which has price less than
$\delta$, then the algorithm will return the endpoint of a tunnel of
price less than or equal to $(1+\eps)3\delta$.  If the algorithm
returns \NULL, then we know that no such tunnel of price less than
$\delta$ exists.

The main idea of the \tunnelTest procedure is the following.  For a
given \tunnel, we can $(1+\eps)$-approximate its price, using a data
structure which answers these queries in polylogarithmic time, see
\dsref{d:s:magic:B}.  The desired tunnel could be a vertical tunnel
which starts at a gate of $\Reached$, or a tunnel between a gate of
$\Reached$ and $\pnt$.  Naively, one could test all tunnels that start
from a gate in $\Reached$ and end in $\pnt$, however, this takes time
at least linear in the size of $\Reached$. Since we are only
interested in a constant factor approximation, it is sufficient, by
\lemref{monotone:shortcut}, to test only the \tunnel which corresponds
to the shortest subcurve of $\cX$.  The corresponding gates can be
found in polylogarithmic time using a two dimensional range tree,
which is built on the set $\Reached$ and we assume is available to
us. We can maintain the range tree during the decision procedure
depicted in \figref{decider:infty:inner}. The technical details are
described in the proof of \lemrefpage{s:y:m:f:s}.  An alternative
solution that uses a balanced binary search tree only is described in
\cite{d-raapg-13}.

\begin{figure}[t]%
    \begin{center}%
        \fbox{\begin{minipage}{0.76\linewidth}%
               \tunnelTest{}$(\Reached, \pnt, \eps, \delta)$
               \begin{algorithmic}[1]
                   \STATE Let $\pntA = (\xPntA,\yPntA)$ point in
                   $\Reached$ with $\max$ value of $\xPntA$,\\
                   ~~~~~~~~~~ such that $\xPntA \leq \xPnt$ and
                   $\yPntA < \yPnt$, where $\pnt=(\xPnt,\yPnt)$.
                   \linelab{range:search}
                   
		   \STATE $\phi = \tunnelPrice{}\pth{\pntA,\pnt}$, see
                   \dsref{d:s:magic:B}. %
                   \linelab{tunnel:compute:phi}
                   \IF { $\phi \leq \constSC\delta$ }
                   \linelab{tunnel:test}
                   
                   \STATE Return $\pnt$ ~~~ ~ \CodeComment{// \tunnel
                      $\xtunnel{\pntA}{\pnt}$} \linelab{X:i}
                   
                   \ENDIF
                   
                   \STATE Compute $j$ such that $\xPnt \in
                   \IEdge{\cX}{j} = [\x_j, \x_{j+1}]$
                   \linelab{col:index}
                   
                   \STATE Let $\pntA = (\xPntA, \yPntA)$ point in
                   $\Reached$ with $\min$ value of $\xPntA$,\\
                   ~~~~~~~~~~ such that $\xPntA \in \IEdge{\cX}{j}$,
                   $\xPntA \geq \xPnt$, and $\yPntA < \yPnt$
                   \linelab{range:search:2}
                   
                   \IF{ $\pntA$ does not exist}%
                   \STATE Return \NULL.
                   \ENDIF%
                   
                   \STATE $\pntD = ( \xPntA , \yPnt)$
                   
                   \IF{ $\frVal{\pntD} \leq \delta$}
                   
                   \STATE Return $\pntD$ \linelab{X:i:2} %
                   ~~~~~\CodeComment{// vertical \tunnel
                      $\xtunnel{\pntA}{\pntD}$}
                   
                   \ELSE
                   
                   \STATE Return \NULL.
                   
                   \ENDIF
                   
               \end{algorithmic}
           \end{minipage}%
        }
    \end{center}%
    \vspace{-0.5cm}%
    \caption{%
       The \tunnelTest procedure receives a set of gates $\Reached$
       and a gate $\pnt$ in the parametric space and returns the
       endpoint of an affordable tunnel between $\Reached$ and $\pnt$
       (or a point close to it) if it exists. The technical details of
       the range queries in \lineref{range:search} and
       \lineref{range:search:2} are described in the proof of
       \lemrefpage{s:y:m:f:s}.  }
    \figlab{tunnel:test}%
\end{figure}%

\subsection{The decision algorithm}
\seclab{algo:decider}

In the decision problem we want to know whether the shortcut \Frechet
distance between two curves, $\cX$ and $\cY$, is smaller or equal a
given value $\delta$.  The free space diagram
$\FullFDleq{\delta}(\cX,\cY)$ may consist of a certain number of
disconnected components and our task is to find a monotone path from
$(0,0)$ to $(1,1)$, that traverses these components while using
shortcuts between vertices of $\cY$ to ``bridge'' between points in
different components or where there is no monotone path connecting
them (see \figref{simple:example}).  The decision algorithm exploits
the monotonicity of the tunnel prices shown in
\lemref{monotone:shortcut} and is based on a breadth first search in
the free space diagram (a similar idea was used
in \cite{dhw-afdrc-12}, but here the details are more involved).

\begin{figure}[t]
    \begin{center}
        \fbox{\begin{minipage}{\FigWidth}%
               \deciderFr{}$(\cX, \cY,\eps, \delta)$
               \begin{algorithmic}[1]
                   \STATE Assert that $\frVal{ 0,0} =
                   \distX{\cX(0)}{\cY(0)} \leq \delta$ and $\frVal{
                      1,1} \leq \delta$
                   
                   \STATE Let $\Queue$ be a min-priority queue for
                   nodes $v(i,j)$ with keys $(jn+i)$
                   
                   \STATE Compute and enqueue the cells
                   $\CellXY{i}{j}$ that have non-empty
                   $\IHCellXY{i}{j}$ or $\IVCellXY{i}{j}$.
                   
                   \STATE Let $\Reached = \brc{(0,0)}$.
                   
                   \WHILE{$\Queue \neq \emptyset$}
                   
                   \STATE Dequeue node $v(i,j)$ and its copies from
                   $\Queue$
                   
                   \STATE Let $\pnt$ be the left gate of
                   $\IHCellXY{i}{j}$
                   
                   \STATE $\pntD = \tunnelTest{}(\Reached, \pnt, \eps,
                   \delta)$ \linelab{decider:t:t}
                   
                   \STATE Compute $\RHCellXY{i}{j}$ and
                   $\RVCellXY{i}{j}$ from $\pntD$,
                   $\RVCellXY{i-1}{j}$, $\RHCellXY{i}{j-1}$,
                   $\IVCellXY{i}{j}$ and $\IHCellXY{i}{j}$
                   
                   \IF{$\RVCellXY{i}{j} \neq \emptyset$}
                   
                   \STATE Enqueue $v(i+1,j)$ and insert edge between
                   $v(i,j)$ and $v(i+1,j)$
                   
                   \ENDIF
                   
                   \IF{$\RHCellXY{i}{j} \neq \emptyset$}
                   
                   \STATE Enqueue $v(i,j+1)$ and insert edge between
                   $v(i,j)$ and $v(i,j+1)$
                   
                   \STATE Add gates of $\RHCellXY{i}{j}$ to $\Reached$
                   
                   \ENDIF
                   
                   
                   
                   
                   \ENDWHILE
                   
                   \IF{$(1,1) \in \Reached$}
                   
                   \STATE Return ``$\distoSFr{\cX}{\cY} \leq
                   (1+\eps)\constSC\delta$''
                   
                   \ELSE
                   
                   \STATE Return ``$\distoSFr{\cX}{\cY} > \delta$''
                   
                   \ENDIF
                   
               \end{algorithmic}
           \end{minipage}}%
    \end{center}
    \vspace{-0.6cm}%
    \caption{The decision procedure \deciderFr{} for the shortcut
       \Frechet distance.  }
    \figlab{decider:infty:inner}
\end{figure}

Given two curves $\cX$ and $\cY$, and parameters $\delta$ and $\eps$,
the algorithm may output an answer equivalent to ``yes'' if there
exists a shortcut curve $\cY'$ of $\cY$, such that
$\distoSFr{\cX}{\cY'} \leq \delta$ and an answer equivalent ``no'' if
there exists no shortcut curve such that $\distoSFr{\cX}{\cY'} \leq
(1+\eps)3\delta$.

\subsubsection{Detailed description of the decision procedure}
\seclab{b:f:s:decider}

The decision algorithm is depicted in \figref{Decider:infty} (and
\figref{decider:infty:inner}).  The algorithm uses a directed graph
$\Graph$ that has a node $v(i,j)$ for every free space cell
$\CellXY{i}{j}$ whose boundary has a non-empty intersection with the
free space $\FullFDleq{\delta}(\cX,\cY)$.  These intersections are
defined as the free space intervals $\IHCellXY{i}{j}$,
$\IVCellXY{i}{j}$, $\IHCellXY{i-1}{j}$ and $\IVCellXY{i}{j-1}$, see
\secref{background}.  For any path along the edges of the graph
$\Graph$ from $v(1,1)$ to $v(i,j)$, there exists a monotone path that
traverses the corresponding cells of $\FullFDleq{\delta}(\cX,\cY)$
while using zero or more affordable tunnels.  A node $v(i,j)$ can have
an incoming edge from another node $v(i',j')$, if $i'\leq i$ and $j'
\leq j$ and either $v(i',j')$ is a neighboring node, or the two cells
can be connected by an affordable tunnel which starts at the lower
boundary of the cell corresponding to $v(i',j')$ and ends at the upper
boundary of the cell corresponding to $v(i,j)$.  The idea of the
algorithm is to propagate reachability intervals $\RVCellXY{i}{j}
\subset \IVCellXY{i}{j}$ and $\RHCellXY{i}{j} \subset \IHCellXY{i}{j}$
while traversing a sufficiently large subgraph starting from $v(1,1)$,
and computing the necessary parts of this subgraph on the fly.  We
store these intervals with the cell $v(i,j)$ that has them on the top
(resp., right) boundary.  The reachability intervals $\RVCellXY{i}{j}$
being computed satisfy
\begin{equation}
    \FDleqI{\infty}{\delta}\pth{\cX,\cY} \cap \IVCellXY{i}{j}%
    \subseteq%
    \RVCellXY{i}{j}%
    \subseteq%
    \FDleqI{\infty}{(1+\eps)3\delta}\pth{\cX,\cY} \cap
    \IVCellXY{i}{j},
    \eqlab{reachable:good}
\end{equation}
and an analogous statement applies to $\RHCellXY{i}{j}$.  The aim is
to determine if either $(1,1) \in
\FDleqI{\infty}{(1+\eps)3\delta}\pth{\cX,\cY}$ or $(1,1) \notin
\FDleqI{\infty}{\delta}\pth{\cX,\cY}$.  Throughout the whole algorithm
we also maintain a set of gates $\Reached$, which represents the
endpoints of the horizontal reachability intervals computed so far.

We will traverse the graph by handling the nodes in a row-by-row
order, thereby handling any node $v(i,j)$ only after we handled the
nodes $v(i',j')$, where $j'\leq j$, $i'\leq i$ and $(i'+j')<(i+j)$.
To this end we keep the nodes in a min-priority queue where the node
$v(i,j)$ has the key $(jn+i)$.  The correctness of the computed
reachability intervals will follow by induction on the order of these
keys.  Furthermore, it will ensure that we handle each node at most
once and that we traverse at most three of the incoming edges to each
node of the graph.

The queue is initialized with the entire node set at once.  To compute
this initial node set and the corresponding free space intervals we
use \dsref{d:s:magic:A}.  The algorithm then proceeds by handling
nodes in the order of extraction from this queue.  When dequeuing
nodes from the queue, the same node might appear three times
(consecutively) in this queue. Once from each of its direct neighbors
in the grid and once from the initial enqueuing.

In every iteration, the algorithm dequeues the one or more copies of
the same node $v(i,j)$ and merges them into one node if necessary.
Assume that $v(i,j)$ has an incoming edge that corresponds to an
affordable tunnel.  Let $\pnt$ be the left gate of
$\IHCellXY{i}{j}$. We invoke \tunnelTest{}$(R,\pnt,\eps,\delta)$ to
test if this is the case.  If the call returns \NULL, then there is no
such affordable tunnel.  Otherwise, we know that the returned point
$\pntD$ is contained in $\RHCellXY{i}{j}$.
If there were more than one copies of this node in the queue, we also
access the reachability intervals of the one or two neighboring
vertices (i.e., $\RVCellXY{i-1}{j}$ and $\RHCellXY{i}{j-1}$).  Using
the reachability information from the at most three incoming edges
obtained this way, we can determine if the cells $\CellXY{i}{j+1}$ and
$\CellXY{i+1}{j}$ are reachable, by computing the resulting
reachability intervals $\RHCellXY{i}{j}$ at the top side and
$\RVCellXY{i}{j}$ the right side of the cell $\CellXY{i}{j}$.  Since
the free space within a cell is convex and of constant complexity,
this can be done in constant time.

Now, if $\RHCellXY{i}{j} \neq \emptyset$ we create a node $v(i,j+1)$,
connect it to $v(i,j)$ by an edge, we enqueue it, and add the gates of
$\RHCellXY{i}{j}$ to $\Reached$.  If $\RVCellXY{i}{j} \neq \emptyset$
we create a node $v(i+1,j)$, connect it to $v(i,j)$ by an edge, and we
enqueue it.  If we discover that the top-right corner of the free
space diagram is reachable this way, we output the equivalent to
``yes'' and the algorithm terminates. In this case we must have added
$(1,1)$ as a gate to $\Reached$.  The algorithm may also terminate
before this happens if there are no more nodes in the queue, in this
case we output that no suitable shortcut curve exists.

\begin{figure}[t]
    \begin{center}
        \fbox{\begin{minipage}{\FigWidth}%
               \DeciderFr{}$(\cXOrig, \cYOrig,\eps, \delta)$
               \begin{algorithmic}[1]
                   \STATE Let $\eps'=\eps/10$
                   
                   \STATE Compute $\cX=\simpX{\cXOrig,\sRadius}$ and
                   $\cY=\simpX{\cYOrig,\sRadius}$ with
                   $\sRadius=\eps'\delta$
	       
                   \STATE Call $\deciderFr(\cX,\cY,\eps',\delta')$
                   with $\delta'=(1+2\eps')\delta$
                   \linelab{decider:replace}
                   
                   \STATE Return either ``$\distoSFr{\cXOrig}{\cYOrig}
                   \leq (1+\eps)\constSC\delta$'' or
                   ``$\distoSFr{\cXOrig}{\cYOrig} > \delta$''
                   
               \end{algorithmic}
           \end{minipage}}%
    \end{center}
    \vspace{-0.6cm}%
    \caption{The resulting decision procedure \DeciderFr{}.  A
       detailed description of the complete algorithm is given in
       \secref{b:f:s:decider}. }
    \figlab{Decider:infty}
\end{figure}

\subsection{The main algorithm}
\seclab{algo:main}

The given input is two curves $\cX$ and $\cY$.  We want to use the
approximate decision procedure \DeciderFr, described above, in a
binary search like fashion to compute the shortcut \Frechet distance.
Conceptually, one can think of the decider as being exact. In
particular, the algorithm would, for a given value of $\delta$, call
the decision procedure twice with parameters $\delta$ and $\delta'=
\delta / \constDec$ (using $\eps=1/3$). If the two calls agree, then
we can make an exact decision, if the two calls disagree, then we can
output a $O(1)$-approximation of the shortcut \Frechet distance.

The challenge is how to choose the right subset of candidate values to
guide this binary search.  Some of the techniques used for this search
have been introduced in previous papers. In particular, this holds for
the search over vertex-vertex, vertex-edge and monotonicity events
which we describe as preliminary computations in \secref{first:stage}.
This stage of the algorithm eliminates the candidate values that also
need to be considered for the approximation of the standard \Frechet
distance and it is almost identical to the algorithm presented in
\cite{dhw-afdrc-12}.

As mentioned before, a monotone path could also become usable by
taking a tunnel.  There are two types of events associated with a
\tunnel family: The first time such that any \tunnel in this family is
feasible, which is the \emphi{creation radius}.  Fortunately, the
creation radii of all \tunnels are approximated by the set of
vertex-vertex and vertex-edge event radii, and our first stage search
(see \secref{first:stage}) would thus take care of such events.

The other events we have to worry about are the first time that the
feasible family of \tunnel{}s becomes usable via a \tunnel (i.e., the
price of some tunnel in this family is below the distance threshold
$\delta$).  Luckily, it turns out that it is sufficient to search over
the price of the canonical \tunnel{} associated with such a family.
The price of a specific \tunnel can be approximated quickly using
\dsref{d:s:magic:B}.  However, there are $\Theta(n^4)$ \tunnel
families, and potentially the algorithm has to consider all of them.
Fortunately, because of $c$-packedness, only $O(n^2)$ of these events
are relevant.  A further reduction in running time is achieved by
using a certain monotonicity property of the prices of these
\tunnel{}s and our ability to represent them implicitly to search over
them efficiently.

\subsubsection{The algorithm -- First stage}
\seclab{first:stage}

We are given two $c$-packed polygonal curves $\cXOrig$ and $\cYOrig$
with total complexity $n$.  We repeatedly compute sets of event values
and perform binary searches on these values as follows.

We compute the set of vertices $V$ of the two curves, and using
well-separated pairs decomposition, we compute, in $O(n \log n)$ time,
a set $U$ of $O(n)$ distances that, up to a factor of two, represents
any distance between any two vertices of $V$. Next, we use \DeciderFr
(with fixed $\eps = 1/3$) to perform a binary search for the atomic
interval in $U$ that contains the desired distance. Let $[\alpha,
\beta]$ denote this interval.  If $10\alpha \geq \beta/10$ then we are
done, since we found a constant size interval that contains the
\Frechet distance. Otherwise, we use the decision procedure to verify
that the desired radius is not in the range $[\alpha, 10\alpha]$ and
$[\beta/10, \beta]$.  For $\alpha' = 3\alpha$, $\beta' = \beta/3$, let
$\Interval'=[\alpha',\beta']$ denote the obtained interval.

We now continue the search using only \deciderFr and the simplified
curves $\cX = \simpX{\cXOrig,\sRadius}$ and $\cY =
\simpX{\cYOrig,\sRadius}$, where $\sRadius = \alpha'$.  We extract the
vertex-edge events of $\cX$ and $\cY$ that are smaller than $\beta'$,
see \secref{background}. To this end, we compute all edges of $\cX$
that are in distance at most $\beta'$ of any vertex of $\cY$ and vice
versa using \dsref{d:s:magic:A}.  Let $U'$ be the set of resulting
distances.  We perform a binary search, using \deciderFr to find the
atomic interval $\Interval''=[\alpha'', \beta'']$ of $U' \cap
\Interval'$ that contains the shortcut \Frechet distance between $\cX$
and $\cY$.

Finally, we again search the margins of this interval, so that either
we found the desired approximation, or alternatively we output the
interval $[10\alpha'', \beta''/10]$,

\subsubsection{Second Stage -- Searching over tunnel prices}
\seclab{second:stage}


It remains to search over the canonical prices of tunnel families
$\xTunnels{\edge}{\edge'}{u}{v}$, where $\edge \neq \edge'$%
\footnote{Since for the case where $\edge=\edge'$ the canonical price
   coincides with the creation event value.}.  After the first stage,
we have an interval $[\alpha,\beta] = [10\alpha'', \beta''/10]$, and
simplified curves $\cX$ and $\cY$ of which the shortcut \Frechet
distance is contained in $[\alpha,\beta]$ and approximates
$\distoSFr{\cXOrig}{\cYOrig}$.  By \lemref{complexity:l:e:q}, the
number of vertex-edge pairs in distance $\beta$ is bounded by
$O(cn/\eps)$.  The corresponding horizontal grid edges in the
parametric space contain the canonical gates which are feasible for
any value in $[\alpha,\beta]$.  Let $\PntSet$ denote the
$m=O(cn/\eps)$ points in the parametric space that correspond to the
canonical gates of these vertex-edge pairs; %
that is, for every feasible pair $\pnt$ (a vertex of $\cY$) and
$\edge$ (an edge of $\cX$), we compute the closest point $\pntA$ on
$\edge$ to $\pnt$, and place the point corresponding to $(\pntA,
\pnt)$ in the free space into $\PntSet$.

It is sufficient to consider the tunnel families between these
vertex-edge pairs, since all other families are not feasible in the
remaining search interval.  Thus, if we did not care about the running
time, we could compute and search over the prices of the tunnels
$\PntSet \times \PntSet$, using \dsref{d:s:magic:B}. Naively, this
would take roughly quadratic time. Instead, we use a more involved
implicit representation of these \tunnel{}s to carry out this task.

\paragraph{Implicit search over \tunnel prices.}

Consider the implicit matrix of tunnel prices $M=\PntSet \times
\PntSet$ where the entry $M(i,j)$ is a $(1+\eps)$-approximation to the
price of the canonical tunnel $\xtunnel{\pnt_i}{\pnt_j}$. By
\lemref{monotone:matrix}, the first $j$ values of the $j$\th row of
this matrix are monotonically decreasing up to a constant factor,
since they correspond to tunnels that share the same endpoint $\pnt_j$
and are ordered by their starting points $\pnt_i$ (we ignore the
values in this matrix above the diagonal).  Using \dsref{d:s:magic:B}
we can $(1+\eps)$-approximate a value in the matrix in polylogarithmic
time per entry.  Similarly, the lower triangle of this matrix is
sorted in increasing order in each column. As such, this matrix is
sorted in both rows and columns and one can apply the algorithm of
Frederickson and Johnson \cite{fj-gsrsm-84} to find the desired
value. This requires $O(\log m)$ calls to \DeciderFr, the evaluation
of $O(m)$ entries in the matrix, and takes $O( m)$ time otherwise.
Here, we are using \DeciderFr as an exact decision procedure.  The
algorithm will terminate this search with the desired constant factor
approximation to the shortcut \Frechet distance.

\section{Analysis}
\seclab{analysis:unbounded}

\subsection{Analysis of the \tunnelTest{} procedure}

\begin{lemma}%
    \lemlab{tunnel:test}%
    Given the left gate $\pnt$ of a free space interval
    $\IHCellXY{i}{j}$ and a set of gates $\Reached$, and parameters
    $0< \eps \leq 1$ and $\delta>0$, the algorithm $\tunnelTest$
    depicted in \figref{tunnel:test} outputs one of the following:
    \begin{compactenum}[(i)]
        \item A point $\pntD \in \IHCellXY{i}{j}$, such that there
        exists a tunnel $\xtunnel{\pntA}{\pntD}$ of price
        $\scPrice{\pntA}{\pntD} \leq (1+\eps)3\delta$ from a gate
        $\pntA \in \Reached$, or
        
        \item \NULL, in this case, there exists no tunnel of price
        less than or equal to $\delta$ between a gate of $\Reached$
        and a point in $\IHCellXY{i}{j}$.
    \end{compactenum}
    Furthermore, in case $(i)$, there exists no other point $\pntB \in
    [\pnt,\pntD]$ that is the endpoint of a tunnel from $\Reached$
    with price less than or equal to $\delta$.
\end{lemma}
\begin{proof}
    The correctness of this procedure follows from the monotonicity of
    the tunnel prices, which is testified by
    \lemref{monotone:shortcut}.  Let $\phi$ be the
    $(1+\eps)$-approximation to the price of the tunnel, that we
    compute in \lineref{tunnel:compute:phi}.  This tunnel starts at a
    point in $\Reached$ and ends in $\pnt$ and it corresponds to the
    shortest subcurve $\widehat{\cX}$ of $\cX$ over any such tunnel.
    \lemref{monotone:shortcut} implies that if $\phi < 3\delta$ then
    there can be no other tunnel of price less than $\delta$, which
    corresponds to a subcurve of $\cX$ that contains $\widehat{\cX}$.
    Therefore, the price of any tunnel from a point $\pntA \in
    \Reached$, which lies in the lower left quadrant of $\pnt$, to a
    point that lies in the upper right quadrant of $\pnt$ has a price
    larger than $\delta$. In particular, this holds for those tunnels
    that end to the right of $\pnt$ in the same free space interval.
    The only other possibility for a tunnel from $\Reached$ to
    $\IHCellXY{i}{j}$ is a vertical tunnel that lies to the right of
    $\pnt$.  Observe that a vertical tunnel which is feasible for
    $\delta$ always has price at most $\delta$, since it corresponds
    to a subcurve of $\cX$ that is equal to a point which is in
    distance $\delta$ to the shortcut edge.  In \lineref{col:index}
    and \lineref{range:search:2} we compute the leftmost gate of
    $\Reached$ in the lower right quadrant of $\pnt$ which lies in the
    same column as $\pnt$. If there exists such a point with a
    vertical tunnel that ends in the free space interval
    $\IHCellXY{i}{j}$, then we return the endpoint of this tunnel.
    Otherwise we can safely output the equivalent to the answer that
    there exists no tunnel of price less than $\delta$.
\end{proof}

\subsection{Analysis of the decision procedure}

Clearly, the priority queue operations take $O(N \log N)$ time and
$O(N)$ space, where $N=\Nleq{\delta}(\cX, \cY)$ is the size of the
node set, which corresponds to the complexity of the free space
diagram.  We invoke the \tunnelTest procedure once for each node.
Since we add at most a constant number of gates for every cell to
$\Reached$, the size of this set is also bounded by $O(N)$.
Therefore, after the initialization the algorithm takes time near
linear in the complexity of the free space diagram.  We can reduce
this complexity by first simplifying the input curves with
$\sRadius=\Theta(\eps\delta)$ before invoking the \deciderFr
procedure, thereby paying another approximation factor. We denote the
resulting wrapper algorithm with \DeciderFr, it is depicted in
\figref{Decider:infty}.  Now, the initial computation of the nodes
takes near-linear time by \dsref{d:s:magic:A} and therefore the
overall running time is near linear.  A more detailed analysis of the
running time can be found in the following.

\begin{lemma}%
    \lemlab{s:y:m:f:s}%
    Given parameters $\delta > 0$ and $0 < \eps \leq 1$ and two
    $c$-packed polygonal curves $\cXOrig$ and $\cYOrig$ {in $\Re^d$}
    of total complexity $n$. The algorithm \DeciderFr depicted in
    \figref{Decider:infty} outputs one of the following:
    \begin{inparaenum}[(i)]
        \item ``$\distoSFr{\cXOrig}{\cYOrig} \leq
        (1+\eps)\constSC\delta$'', or
	\item ``$\distoSFr{\cXOrig}{\cYOrig} > \delta$''.
    \end{inparaenum}
    In any case, the output returned is correct. The running time is
    $O\pth{ C n \log^2 n }$, where $C = c^2\eps^{-2d} \log (1/\eps)$.
\end{lemma}

\begin{proof}
    The algorithm $\DeciderFr$ computes the simplified curves $\cX =
    \simpX{\cXOrig,\sRadius}$ and $\cY = \simpX{\cYOrig,\sRadius}$
    with $\sRadius=\Theta(\eps\delta)$, before invoking the algorithm
    \deciderFr described in \figref{decider:infty:inner} on these
    curves.  By the correctness of the \tunnelTest{} procedure (i.e.,
    \lemref{tunnel:test}), one can argue by induction that the subsets
    of points of $\FDleqI{\infty}{\delta}\pth{\cX,\cY}$ intersecting a
    grid edge are sufficiently approximated by the reachable intervals
    computed by \deciderFr (see \Eqrefpage{reachable:good}).  By
    \lemref{t:r:i:n:e:q:s:u:b}, this approximates the decision with
    respect to the original curves sufficiently.
    
    It remains to analyze the running time. By
    \lemref{complexity:l:e:q}, the size of the node set of the graph
    $\Graph$ is bounded by $N = O(cn/\eps)$.  This also bounds the
    size of the point set $\Reached$ and the number of calls to the
    \tunnelTest procedure, as those are at most a constant number per
    node.  During the \tunnelTest procedure, which is depicted in
    \figref{tunnel:test}, we%
    \smallskip%
    \begin{compactenum}[\;\;\;(A)]
        \item approximate the price of one \tunnel in
        \lineref{tunnel:compute:phi}, and
        \item invoke two orthogonal range queries on the set
        $\Reached$ in \lineref{range:search} and
        \lineref{range:search:2}.
    \end{compactenum}
    \smallskip%
    As for (A), building the data structure that supports this kind of
    queries takes $T_1 = O\pth{ n \eps^{-2d} \log^2 (1/\eps) \log^2 n
    }$ time by \dsref{d:s:magic:B}. Since we perform $O(N)$ such
    queries, this takes $T_2 = O\pth{ N \eps^{-3} \log n \log \log n}
    = O\pth{ c n \eps^{-4} \log n \log \log n }$ time overall.  As for
    (B), again, the set of gates $\Reached$ is a finite set of two
    dimensional points and we can use two dimensional range trees
    (with fractional cascading as described in \cite{bcko-cgaa-08}) to
    support the orthogonal range queries. We want to build this tree
    by adding $O(N)$ points throughout the algorithm execution.  Since
    the range tree is a static data structure, we have to make it
    dynamic, but we only need to support insertions, and no deletions.
    This can be easily done by using the logarithmic method if we
    allow an additional logarithmic factor to the running time, see
    also \cite{bs-dspsd-80, o-ddds-83}.  In this method, the point set
    is distributed over $O(\log N)$ static range trees, which need to
    be queried independently and which are repeatedly rebuilt
    throughout the algorithm.  Overall, maintaining this data
    structure and answering the orthogonal range queries takes $T_3 =
    O( N \log^2 N )$ time.
    
    During the algorithm, we maintain a priority queue, where each
    node is added and extracted at most three times. As such, the
    priority queue operations take time in $O(N \log N)$.  The initial
    computation of the node set takes $T_4=O(n \log n + c^2n / \eps )$
    by \dsref{d:s:magic:A}.
    
    Therefore, the overall running time is $T_1 + T_2 + T_3 + T_4$,
    which is
    %
    \begin{align*}
        &O\pth{ n \eps^{-2d} \log^2 (1/\eps) \log^2 n 
           + %
           c n \eps^{-4} \log n \log \log n 
           + %
           c n \log^2 n 
           + %
           n \log n + c^2n / \eps 
        }
        \\
        & \hspace*{2cm}=%
        O\pth{ C n \log^2 n },
    \end{align*}%
    where $C = c^2\eps^{-2d} \log (1/\eps)$.
\end{proof}

\begin{observation}
    \obslab{reparametrizations}%
    It is easy to modify the \deciderFr algorithm such that it also
    outputs the respective shortcut curve and reparametrization which
    satisfies the \Frechet distance. We would modify the tunnel
    procedure such that it returns not only the endpoint, but also the
    starting point of the computed tunnel. During the algorithm, we
    then insert an edge for each computed tunnel, thereby creating at
    most three incoming edges to each node.  After the algorithm
    terminates, we can trace any path backwards from $(1,1)$ to
    $(0,0)$ in the subgraph computed this way. This path encodes the
    shortcut curves as well as the reparametrizations.%
\end{observation}%

\subsection{Analysis -- understanding \tunnel events}
\seclab{shortcut:price}

The main algorithm uses the procedure \DeciderFr to perform a binary
search for the minimum $\delta$ for which the decision procedure
returns ``yes''.  In the problem at hand we are allowed to use
\tunnels to traverse the free space diagram, and it is possible that a
path becomes feasible by introducing a \tunnel.  The algorithm has to
consider this new type of critical events.


Consider the first time (i.e., the minimal value of $\delta$) that a
decision procedure would try to use a \tunnel of a certain family.

\begin{defn}%
    \deflab{r:create}%
    Given a \tunnel family $\xTunnels{\edge_i}{\edge_j}{u}{v}$, we
    call the minimal value of $\delta$ such that
    $\xTunnelsLeq{\edge_i}{\edge_j}{u}{v}{\delta}$ is non-empty the
    \emphi{creation radius} of the \tunnel family and we denote it
    with $\rCreate{\edge_i}{\edge_j}{u}{v}$.  (Note, that the price of
    a \tunnel might be considerably larger than its creation radius.)
\end{defn}

\begin{lemma}%
    \lemlab{r:create}%
    The creation radius $\rCreate{\edge_i}{\edge_j}{u}{v} =
    \rMinCreate{\edge_i}{\edge_j}{u}{v}$, see \defref{min:radius}.
\end{lemma}

\begin{proof}
    Recall that the creation radius of the \tunnel family is the
    minimal value of $\delta$ such that any \tunnel in this family is
    feasible. Let $u'$ be the closest point on $\edge_i$ to $u$, and
    $v'$ the closest point on $\edge_j$ to $v$.  If $u'$ appears
    before $v'$ on $\cX$, then the canonical \tunnel is realized by
    $\cX(\xPntA)=u'$ and $\cX(\xPntA)=v'$ and the claim holds. In
    particular, this is the case if $i < j$.
    
    Now, the only remaining possibility is that $u'$ appears after
    $v'$ on $\edge$. It must be that $i=j$, therefore let
    $\edge=\edge_i=\edge_j$.  Observe that in this case any \tunnel in
    the family which is feasible for $\delta$ also has a price that is
    smaller or equal to $\delta$.
    Consider the point $\pntB$ realizing the quantity
    \begin{align*}
        \min_{\pntB \in \edge} \max \pth{\, \MakeSBig\!
           \distX{\pntB}{u} ,\, \distX{\pntB}{v} \,}.
    \end{align*}
    Note that $\pntB$ is the subcurve of $\cX$ corresponding to the
    (vertical) canonical \tunnel in this case.  We claim that for any
    subsegment $\wu \wv \subseteq \edge$ (agreeing with the
    orientation of $\edge$) we have that $\distFr{\wu \wv}{uv} \geq
    \distFr{\pntB}{uv}$.  If $\wu = \wv$ then the claim trivially
    holds.
    
    \begin{figure}[\si{tb}]\center
        \includegraphics[scale=0.8]{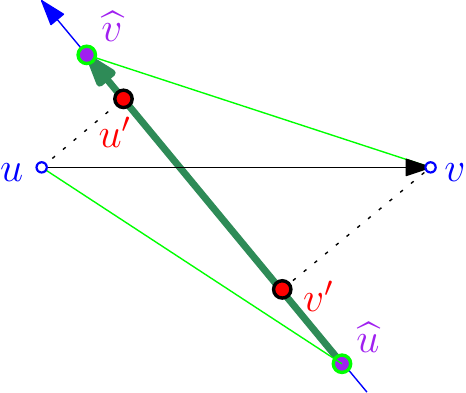}
        ~~~~~~%
        \includegraphics[scale=0.8]{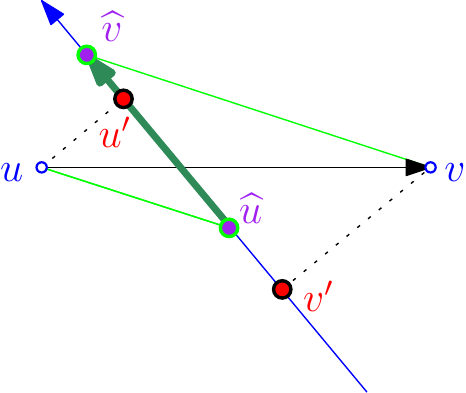}
        \caption{Two cases: $v'$ appears either before or after $\wu$
           along $\edge$, assuming that $u'$ appears after $v'$ on
           $\edge$.}
        \figlab{stupid}
    \end{figure}
    
    Assume that $v'$ appears after $\wu$ along $\edge$ (the case
    depicted in \figref{stupid}). Since $u'$ appears after $v'$ along
    $\edge$, we have that $\distX{v'}{u} \leq \distX{\wu}{u}$, as
    moving away from $u'$ only increases the distance from
    $u$. Therefore,
    \begin{align*}
        \distFr{\pntB}{uv} %
        &\leq \distFr{v'}{uv} = \max\pth{\distX{v'}{u},
           \distX{v'}{v}}%
        \leq%
        \max\pth{\distX{\wu}{u}, \distX{\wv}{v}}%
        =%
        \distFr{\wu \wv}{uv}.
    \end{align*}
    
    Otherwise, if $v'$ appears before $\wu$ along $\edge$, as depicted
    in \figref{stupid} on the right, then
    \begin{align*}
        \distFr{\pntB}{uv} %
        \leq%
        \distFr{\wu}{uv} = \max\pth{\distX{\wu}{u}, \distX{\wu}{v}}%
        \leq%
        \max\pth{\distX{\wu}{u}, \distX{\wv}{v}} %
        =%
        \distFr{\wu \wv}{uv},
    \end{align*}
    since moving away from $v'$ only increases the distance from $v$.
    
    This implies that the minimum $\delta$ for a \tunnel in
    $\xTunnels{\edge_i}{\edge_i}{u}{v}$ to be feasible is at least
    $\distFr{\pntB}{uv} = \rCreate{\edge_i}{\edge_i}{u}{v}$.  And
    $\pntB$ testifies that there is a \tunnel in this family that is
    feasible for this value.
\end{proof}

The following lemma describes the behavior when $\delta$ rises above a
\tunnel price, such that the area in the free space that lies beyond
this \tunnel potentially becomes reachable by using this \tunnel. More
specifically, it implies that the first time (i.e., the minimal value
of $\delta$) that any \tunnel of a family
$\xTunnels{\edge_i}{\edge_j}{u}{v}$ is usable (i.e., its price is less
than $\delta$), any \tunnel in the feasible set
$\xTunnelsLeq{\edge_i}{\edge_j}{u}{v}{\delta}$ associated with this
family will be usable.
\begin{lemma}%
    \lemlab{tunnel:event}%
    Given a value $\delta \geq 0$, we have for any \tunnel
    $\xtunnel{\pntF}{\pntG}$ in the feasible subset of a given \tunnel
    family $\xTunnelsLeq{\edge_i}{\edge_j}{u}{v}{\delta}$, that
    \begin{compactenum}[(i)]
        \item if $\delta \leq
        \priceX{\bCanonical{\edge_i}{\edge_j}{u}{v}}$, then
        $\scPrice{\pntF}{\pntG} =
        \priceX{\bCanonical{\edge_i}{\edge_j}{u}{v}}$,
        
        \item otherwise, $\scPrice{\pntF}{\pntG} \leq \delta$.
    \end{compactenum}
\end{lemma}
\begin{proof}
    We first handle the case that $i \neq j$.  Let $\edge_i = \pnt_i
    \pnt_{i+1}$ and $\edge_j = \pnt_j \pnt_{j+1}$.
    
    \parpic[r]{\includegraphics{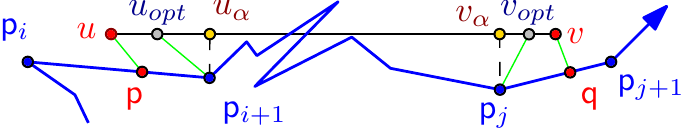}}

    Let $\pnt \in \edge_i$ and $\pntA \in \edge_j$ be some points on
    these edges, that correspond to $\pntF$ and $\pntG$, respectively.
    Observe that since this is a feasible \tunnel in this family, we
    have that
    \begin{align*}
        \max \pth{ \distX{\pnt}{u}, \distX{\pntA}{v} } \leq \delta.
    \end{align*}
    Consider the optimal \Frechet matching of $\SC{\cX}{\pnt}{\pntA}$
    with $u v$, and let $u_\opt$ and $v_\opt$ be the points on $uv$
    that are matched to $\pnt_{i+1}$ and $\pnt_{j}$ by this optimal
    \Frechet matching.  Let $\alpha =
    \distFr{\SC{\cX}{\pnt_{i+1}}{\pnt_j}}{u_\alpha v_\alpha}$, where
    $u_\alpha v_\alpha$ is the subsegment of $uv$ minimizing $\alpha$.
    
    We have, by \obsref{f:r:segments}, that
    \begin{align*}
        \distFr{\SC{\cX}{\pnt}{\pntA}}{uv} %
        &=%
        \max\pth{\, \distFr{\pnt \pnt_{i+1}}{u {} u_\opt},\,%
           \distFr{ \MakeSBig \SC{\cX}{\pnt_{i+1}}{\pnt_j}}{u_{\opt}
              v_{\opt}},\,%
           \distFr{\pnt_{j} \pntA}{v_\opt v}\, }
        \\
        &=%
        \max {\pth{\MakeSBig\!%
              \begin{array}{c}
                  \distX{\pnt}{u},\\
                  \distX{\pnt_{i+1}}{u_\opt},\\
              \distFr{\SC{\cX}{\pnt_{i+1}}{\pnt_j}}{u_{\opt}
                 v_{\opt}},\\
              \distX{\pnt_{j}}{v_\opt},\\
              \distX{\pntA}{v}
              \end{array}
           }}
        \\
        &= \max {\pth{\, \MakeSBig\!%
              \distX{\pnt}{u},\,%
              \distFr{\SC{\cX}{\pnt_{i+1}}{\pnt_j}}{u_{\opt}
                 v_{\opt}},\,%
              \distX{\pntA}{v} \, }}
        \\
        &\geq%
        \max \pth{\,%
           \distX{\pnt}{u}, \,%
           \distFr{\SC{\cX}{\pnt_{i+1}}{\pnt_{j}}} {u_\alpha
              v_\alpha},\,%
           \distX{\pntA}{v}\,%
        }%
        \\
        &=%
        \max\pth{\, \MakeSBig%
           \distFr{\pnt \pnt_{i+1}}{u {} u_\alpha},\,%
           \distFr{\SC{\cX}{\pnt_{i+1}}{\pnt_{j}}} {u_\alpha
              v_\alpha},\,%
           \distFr{\pnt_{j} \pntA}{v_1 v}\, }
        \\
        &\geq%
        \distFr{\SC{\cX}{\pnt}{\pntA}}{uv}.
    \end{align*}%
 
    For $\alpha = \distFr{\SC{\cX}{\pnt_{i+1}}{\pnt_{j}}} {u_\alpha
       v_\alpha}$, this implies $\distFr{\SC{\cX}{\pnt}{\pntA}}{uv} =
    \max ( \distX{\pnt}{u}, \alpha, \distX{\pntA}{v})$, where $\alpha
    \leq \max( \alpha, \delta)$ is equal for all \tunnels in the
    family.  Now, if $\delta \leq$ $
    \priceX{\bCanonical{\edge_i}{\edge_j}{u}{v}}$ then we have
    $\priceX{\bCanonical{\edge_i}{\edge_j}{u}{v}}=\alpha \geq \delta$
    and
    \begin{align*}
        \scPrice{\pntF}{\pntG} = \distFr{\SC{\cX}{\pnt}{\pntA}}{uv} =
        \max(\distX{\pnt}{u}, \alpha, \distX{\pntA}{v}) \leq
        \max(\alpha,\delta) =\alpha.
    \end{align*}%
    This proves (i).  Otherwise, we have
    $\priceX{\bCanonical{\edge_i}{\edge_j}{u}{v}} < \delta$. Which
    implies that $\alpha < \delta$, but then $\scPrice{\pntF}{\pntG}
    \leq \delta$, implying (ii).
    
    If $i = j$ then the \Frechet distance is between the shortcut
    segment and a subsegment of $\edge_i$. But this distance is the
    maximum distance between the corresponding endpoints, by
    \obsref{f:r:segments}. As the distance between endpoints of
    shortcuts and subcurves corresponding to \tunnels of
    $\xTunnelsLeq{\edge_i}{\edge_j}{u}{v}{\delta}$ is at most
    $\delta$, and by \lemref{r:create} the claim follows.
\end{proof}


\lemref{monotonicity} below implies that the set of creation radii of
all \tunnels is approximated by the set of vertex-vertex and
vertex-edge event radii.  A similar lemma was shown in
\cite{dhw-afdrc-12}, to prove this property for the monotonicity event
values. Therefore, the algorithm eliminates these types of events in
the first stage, in addition to eliminating the vertex-vertex and
vertex-edge events.

\begin{lemma}%
    \lemlab{monotonicity}%
    Consider an edge $\edge = \pnt \pntA$ of a curve $\cXOrig$, and
    two vertices $u$ and $v$ of a curve $\cYOrig$. We have that $x/2
    \leq \rCreate{\edge}{\edge}{u}{v} \leq 2x$, where $x$ is in the
    set $\brc{ \distSet{u}{\edge}, \distSet{v}{\edge}, \distX{u}{v}}$.
\end{lemma}
\begin{proof}
    First, observe that $\rCreate{\edge}{\edge}{u}{v} \geq
    \distX{u}{v}/2$, as it is the maximum distance of some point on
    $\edge$ from both $u$ and $v$. In particular, if
    $\rCreate{\edge}{\edge}{u}{v} \leq 2 \distX{u}{v}$ then we are
    done.
    
    As such, it must be that $\rCreate{\edge}{\edge}{u}{v} > 2
    \distX{u}{v}$. Assume that $u$ is closer to $\edge$ than $v$, and
    let $u'$ be the closest point on $\edge$ to $u$. By the triangle
    inequality, the distance of $v$ from $u'$ is in the range
    $\Interval = \pbrc{ \distX{u}{u'}, \distX{u}{u'} + \distX{u}{v}
    }$.  Observe that $\rCreate{\edge}{\edge}{u}{v} \geq
    \distX{u}{u'}$ and $\rCreate{\edge}{\edge}{u}{v} \leq \max
    \pth{\distX{u}{u'}, \distX{v}{u'}}$.  Thus,
    $\rCreate{\edge}{\edge}{u}{v} \in \Interval$. Note that if
    $\distX{u}{u'} \leq \distX{u}{v}$ then we are done, as this
    implies that $\rCreate{\edge}{\edge}{u}{v}$ is in the range
    $\pbrc[]{ \distX{u}{v}, 2 \distX{u}{v}}$. Otherwise,
    $\rCreate{\edge}{\edge}{u}{v}$ is in the range $\pbrc[]{
       \distX{u}{u'}, 2 \distX{u}{u'}}$. In either case, the claim
    follows.
    
    The case that $v$ is closer to $\edge$ than $u$ follows by
    symmetry.
\end{proof}

\subsection{Analysis of the main algorithm}
\seclab{algo:main:analysis}

The following lemma can be obtained using similar arguments as in the
analysis of the main algorithm in \cite{dhw-afdrc-12}. We provide a
simplified proof for the case here, where we are only interested in a
constant factor approximation.

\begin{lemma}%
    \lemlab{first:stage}%
    Given two $c$-packed polygonal curves $\cXOrig$ and $\cYOrig$ in
    $\Re^d$ with total complexity $n$, the first stage of the
    algorithm (see \secref{first:stage}) outputs one of the following:
    \begin{compactenum}[(A)]
        \item a $O(1)$-approximation to the shortcut \Frechet distance
        between $\cXOrig$ and $\cYOrig$;
        \item an interval $\widehat{\Interval}$, and curves $\cX$ and
        $\cY$ with the following properties:
        \begin{compactenum}[(i)]
            \item $\distoSFr{\cX}{\cY}$ is contained in
            $\widehat{\Interval}$ and $\distoSFr{\cX}{\cY}/3 \leq
            \distoSFr{\cXOrig}{\cYOrig} \leq 3\distoSFr{\cX}{\cY}$,
            
            \item $\widehat{\Interval}$ contains no vertex-edge,
            vertex-vertex, or monotonicity event values and no \tunnel
            creation radii (as defined in \secref{shortcut:price}) of
            $\cX$ and $\cY$.
        \end{compactenum}
    \end{compactenum}
    The running time is $O\pth{ c^2 n \log^3 n}$.
\end{lemma}

\begin{proof}
    We first prove the correctness of the algorithm as stated in the
    claim.  The set $U$ approximates the vertex-vertex distances of
    the vertices of $\cXOrig$ and $\cYOrig$ up to a factor of
    two. Therefore, the interval $\Interval=[\alpha,\beta]$, which we
    obtain from the first binary search, contains no vertex-vertex
    distance of $\cXOrig$ that is more than a factor of two away from
    its boundary.  This implies that the simplification
    $\cX=\simpX{\cXOrig,\sRadius}$ results in the same curve for any
    $\mu \in [3\alpha,\beta/3]$.  An analogous statement holds for
    $\cYOrig$. Unless, a constant factor approximation is found either
    in the interval $[\alpha, 10\alpha]$ or the interval $[\beta/10,
    \beta]$, the algorithm continues the search using the procedure
    \deciderFr and the curves simplified with $\mu=3\alpha$.

    It is now sufficient to search for a constant factor approximation
    to $\distoSFr{\cX}{\cY}$ in the interval
    $\Interval'=[3\alpha,\beta/3]$, since this will approximate the
    desired \Frechet distance by a constant factor.  Indeed, by the
    result of the initial searches, we have that $3\mu \leq 10 \alpha
    \leq \distoSFr{\cXOrig}{\cYOrig}$.  \lemref{t:r:i:n:e:q:s:u:b}
    imply that $\distoSFr{\cX}{\cY} \leq \distoSFr{\cXOrig}{\cYOrig} +
    2\mu \leq 3 \distoSFr{\cXOrig}{\cYOrig}.$ On the other hand, the
    same lemma implies that $\distoSFr{\cX}{\cY} \geq
    \distoSFr{\cXOrig}{\cYOrig} - 2\mu \geq
    \distoSFr{\cXOrig}{\cYOrig}/3.$ This implies, that
    $\distoSFr{\cX}{\cY} \in \Interval'=[3\alpha,\beta/3]$, since
    $\distoSFr{\cXOrig}{\cYOrig} \in [10\alpha,\beta/10]$.  Note that
    this also proves the correctness of $(i)$, since the returned
    interval is contained in $\Interval'$.
    
    Observe that the set of vertex-vertex distances of $\cX$ and $\cY$
    is contained in the set of vertex-vertex distances of $\cXOrig$
    and $\cYOrig$. Clearly, $\Interval'$ cannot contain any
    vertex-vertex distances of $\cX$ and $\cY$.  The algorithm
    therefore extracts the remaining vertex-edge events $U'$ from the
    free space diagram and performs a binary search on them.  We
    obtain the atomic interval $\Interval''=[\alpha'',\beta'']$, which
    contains no vertex-edge events of $\cX$ and
    $\cY$. 
    Note that by \Eqrefpage{c:tunnel:endpoints} and \lemref{r:create},
    the monotonicity event values, as described in
    \secref{background}, coincide with the values of $\delta$ where a
    tunnel within a column of the parametric space becomes feasible,
    that is, with the quantity $\rCreate{\edge}{\edge}{u}{v}$.  By
    \lemref{monotonicity}, these event values would have to lie within
    a factor two of the boundaries of the interval $\Interval''$.
    Therefore, we again search the margins of this interval, so that
    either we found the desired approximation, or alternatively, it
    must be in the interval $\Interval'''=[10\alpha'', \beta''/10]$,
    which now contains no vertex-vertex, vertex-edge, monotonicity or
    tunnel creation events of $\cX$ and $\cY$.  Since $\Interval'''$
    is the interval that the algorithm returns, unless it finds a
    constant factor approximation to the desired \Frechet distance,
    the above argumentation implies $(i)$ and $(ii)$.
    
    As for the running time, computing the set $U$ using
    well-separated pairs decomposition can be done in $O(n \log n)$
    time, see \cite{dhw-afdrc-12}.  Computing the set $U'$ takes time
    in $O( n \log n + c^2n )$, by \dsref{d:s:magic:A} with
    $\mu=\beta/3$ and $\delta=\beta$.  The algorithm invokes the
    decision procedure $O(\log n)$ times, and this dominates the
    overall running time, see \lemref{s:y:m:f:s}.
\end{proof}

\begin{lemma}%
    \lemlab{rand:algo}%
    Given two $c$-packed polygonal curves $\cXOrig$ and $\cYOrig$ in
    $\Re^d$ of total complexity $n$, one can compute a constant factor
    approximation to $\distoSFr{\cXOrig}{\cYOrig}$.  The running time
    is $O\pth{ c^2 n\log^3 n}$.
\end{lemma}

\begin{proof}
    First, the algorithm performs the preliminary computations as
    described in \secref{first:stage}. By \lemref{first:stage}, we
    either find a constant factor approximation, or we obtain an
    interval $[\alpha,\beta]$ and simplified curves $\cX$ and
    $\cY$. Furthermore, the interval $[\alpha,\beta]$ does not contain
    any vertex-vertex, vertex-edge, monotonicity, or tunnel creation
    events of $\cX$ and $\cY$.  Let $\PntSet$ be the canonical gates
    that are feasible in the \mbox{$\beta$-free} space of $\cX$ and
    $\cY$. We have that $m = \cardin{\PntSet}=O(n)$ and we can compute
    them using \dsref{d:s:magic:A} in $O( n \log n + c^2n)$ time, for
    $\eps=1/3$.  Thus, the running time up to this stage is bounded by
    $O\pth{ c^2 n\log^3 n}$, by \lemref{first:stage}.
    
    Now, we invoke the second stage of the algorithm described in
    \secref{second:stage} on the matrix of implicit tunnel prices
    defined by $\PntSet$ and return the output as our solution.
    
    Consider a monotone path in the parametric space that corresponds
    to the optimal solution. If the price of this path is determined
    by either a vertex-vertex, a vertex-edge or a monotonicity event
    then we have found an approximation to the shortcut \Frechet
    distance already in the first stage of the search algorithm.  If
    it is dominated by a tunnel price and this tunnel has both
    endpoints in the same column of the free space, then by
    \obsref{f:r:segments} it is a creation radius. By
    \lemref{r:create} this is equivalent to the minimum radius of the
    corresponding tunnel family. By \obsref{min:radius:eq} the minimum
    radius corresponds to either a vertex-edge event or a monotonicity
    event. Thus, it lies outside the interval $[\alpha,\beta]$, since
    by \lemref{first:stage} these critical values were eliminated in
    the first stage. Otherwise, this critical tunnel has to be between
    two columns.  Let $\delta$ be the price of this \tunnel (which is
    also the price of the whole solution).
    
    Consider what happens to this path if we slightly decrease
    $\delta$. Since $\delta$ is optimal, then the critical tunnel
    either ceases to be feasible or its price is not affordable
    anymore.
    
    If the critical tunnel is no longer feasible, then one of its
    endpoints is also an endpoint of the free space interval it lies
    on. Consider the modified path in the free space, which uses the
    new endpoint of the free space interval.  If the free space
    interval is empty, then this corresponds to a vertex edge event,
    and this is not possible inside the interval $[\alpha, \beta]$.
    The other possibility is that the path is no longer
    monotone. However, this corresponds to a monotonicity event, which
    again we already handled because of \lemref{first:stage}.
    
    If the tunnel is still feasible, then it must be that the
    endpoints of this \tunnel are contained in the interior of the
    free space interval and not on its boundary.  Now
    \lemref{tunnel:event} (i) implies that the price of this \tunnel
    is equal to the price of the canonical tunnel.  As such, the price
    of the optimal solution is being approximated correctly in this
    case.

    Observe that in the second stage we are searching over all tunnel
    events that lie in the remaining search interval (whether they are
    relevant in our case or not).  Hence, the search would find the
    correct critical value, as it is one of the values considered in
    the search.
    
    The running time of second stage is bounded by:
    \begin{compactenum}[(A)]
        \item $O(n \log n \log \log n)$ time to compute the needed
        entries in the matrix, using \dsref{d:s:magic:B}.
        
        \item $O\pth{ \pth[]{c^2 n \log^2 n} \log n}$ time for the
        $O(\log n)$ calls to \DeciderFr.
        
        \item $O(n)$ for other computations.
    \end{compactenum}
    
    Therefore, the overall running time of the algorithm is $O\pth{
       c^2 n\log^3 n}$.
\end{proof}

\subsection{Result}

The following theorem states the main result for approximating the
shortcut \Frechet distance.

\begin{theorem}%
    \thmlab{main}%
    Given two $c$-packed polygonal curves $\cXOrig$ and $\cYOrig$ in
    $\Re^d$, with total complexity $n$, and a parameter $\eps > 0$,
    the algorithm of \secref{algo:unbounded:general} computes a
    $(3+\eps)$-approximation to the shortcut \Frechet distance between
    $\cX$ and $\cY$ in $\RTUnbounded$ time.  The algorithm also
    outputs the shortcut curve of $\cYOrig$ and the reparametrizations
    that realize the respective shortcut \Frechet distance.
\end{theorem}

\begin{proof}
    The result follows from \lemref{rand:algo}.  This yield an
    interval $\Interval$ that contains the value of the optimal
    solution. We can turn any constant factor approximation into a
    $(3+\eps)$-approximation, using \DeciderFr with $\eps'=\eps/3$ by
    invoking it over a constant number of subintervals of the form
    $[\alpha,\beta]$, where $\beta=(3+\eps)\alpha$. These intervals
    are required to cover $\Interval$, and as such, \DeciderFr would
    return the desired approximation for one of them (the running time
    of each call to \DeciderFr is stated in \lemref{s:y:m:f:s}).
    
    It is easy to modify the algorithm, such that it also outputs the
    shortcut curve and the reparametrizations realizing the
    approximate \Frechet distance, see \obsref{reparametrizations}.
\end{proof}

\begin{remark}%
    \remlab{k:shortcut}%
    One can extend the algorithm of \thmref{main} so that it
    approximates the \Frechet distance where only $k$ shortcuts are
    allowed. The basic algorithm is similar, except that we keep track
    for the points of $\Reached$ how many shortcuts were used in
    computing them. The resulting algorithm has running time
    $\RTBoundedK$ (for $\eps$ a constant).  This version of the
    algorithm is described in the first author's thesis, see
    \cite{d-raapg-13}.
\end{remark}


\section{Data structures for \Frechet-distance queries}
\seclab{single:segment:query} 

Given a polygonal curve $\cZ$ in $\Re^d$, we build a data structure
that supports queries for the \Frechet distance of subcurves of $\cZ$
to query segments $\pnt\pntA$. We describe the data structure in three
stages.  After establishing some basic facts in \secref{query:helper},
we first describe a data structure that achieves a constant factor
approximation in \secref{queries:stage:one}.  We proceed by describing
a data structure that answers queries for the \Frechet distance of the
entire curve to a query segment up to an approximation factor of
$(1+\eps)$ in \secref{queries:stage:two}.  Finally, we describe how to
combine these two results to obtain the final data structure for
segment queries in \secref{queries:stage:three}.

\subsection{Useful lemmas for curves and segments}
\seclab{query:helper}

\begin{defn}\deflab{spine}
    For a curve $\cZ$, the segment connecting its endpoints is its
    \emphi{spine}, denoted by $\spineX{\cZ}$.
\end{defn}

The following is a sequence of technical lemmas that we need later
on. These lemmas testify that:
\begin{compactenum}[\quad (A)]
    \item The spine of a curve is, up to a factor of two, the closest
    segment to this curve with respect to the \Frechet distance, see
    \lemref{f:r:basic}.
    
    \item The \Frechet distance between a curve and its spine is
    monotone, up to a factor of two, with respect to subcurves, see
    \lemref{monotone:subcurve}.
    
    \item Shortcutting a curve cannot increase the \Frechet distance
    of the curve to a line segment, see \lemref{shortcut:segment} or
    \cite{bbw-cfdsp-08}.
\end{compactenum}

\begin{lemma}%
    \lemlab{f:r:basic}%
    Let $\pnt \pntA$ be a segment and $\cZ$ be a curve.  Then,
    \begin{inparaenum}[(i)]
        \item \itemlab{shortcut} $\distFr{\pnt \pntA}{\cZ} \geq$
        $\distFr{\pnt \pntA}{\spineX{\cZ}}$, and
        \item $\distFr{\pnt \pntA}{\cZ} \geq
        \distFr{\spineX{\cZ}}{\cZ}/2$.
    \end{inparaenum}
\end{lemma}

\begin{proof}
    Let $\pntB$ and $\pntC$ be the endpoints of $\cZ$; that is
    $\spineX{\cZ} = \pntB \pntC$.
    
    (i) Since in any \matching{} of $\pnt \pntA$ with $\cZ$ it must be
    that $\pnt$ is matched to $\pntB$, and $\pntA$ is matched to
    $\pntC$, it follows that $\distFr{\pnt \pntA}{\cZ} \geq \max (
    \distX{\pnt}{u}, \distX{\pntA}{v}) = \distFr{\pnt \pntA}{\pntB
       \pntC} = \distFr{\pnt \pntA}{\spineX{\cZ}}$, by
    \obsref{f:r:segments}.
    
    (ii) By (i) and the triangle inequality, we have that
    \[\distFr{\spineX{\cZ}}{\cZ} \leq \distFr{\spineX{\cZ}}{\pnt
       \pntA} + \distFr{\pnt \pntA}{\cZ} \leq 2 \distFr{\pnt
       \pntA}{\cZ},\] which implies the claim.
\end{proof}

\begin{lemma}%
    \lemlab{monotone:subcurve}%
    Given two curves $\cZ$ and $\cubcZ$, such that $\cubcZ$ is a
    subcurve of $\cZ$.  Then, we have that
    $\distFr{\spineX{\cubcZ}}{\cubcZ} \leq 2
    \distFr{\spineX{\cZ}}{\cZ}$.
\end{lemma}

\begin{proof}
    Consider the \matching{} that realizes the \Frechet distance
    between $\cZ$ and $\spineX{\cZ}$. It has to match the endpoints of
    $\cubcZ$ to points $\pntA$ and $\pntB$ on $\spineX{\cZ}$.  We have
    that $\distFr{\cubcZ}{{\pntA\pntB}} \leq
    \distFr{\cZ}{\spineX{\cZ}}$.  By \lemref{f:r:basic} (i), we have
    $\distFr{\spineX{\cubcZ}}{{\pntA\pntB}} \leq
    \distFr{\cubcZ}{{\pntA\pntB}} \leq \distFr{\cZ}{\spineX{\cZ}}$.
    Now, by the triangle inequality, we have that
    \begin{align*}
        \distFr{\cubcZ}{\spineX{\cubcZ}} \leq
        \distFr{\cubcZ}{{\pntA\pntB}} +
        \distFr{{\pntA\pntB}}{\spineX{\cubcZ}} \leq 2
        \distFr{\cZ}{\spineX{\cZ}}.
    \end{align*}
    \aftermathA
\end{proof}

\begin{lemma}
    \lemlab{shortcut:segment}%
    Let $\cZ = u_1 u_2 \ldots u_n$ be a polygonal curve, $\pnt \pntA$
    be a segment, and let $i < j$ be any two indices. Then, for $\cZ'
    = \SC{\cZ}{u_1}{u_i} \concatOp u_i u_j \concatOp
    \SC{\cZ}{u_j}{u_n}$, we have $\distFr{\cZ'}{ \pnt \pntA} \leq
    \distFr{\cZ}{ \pnt \pntA}$.
\end{lemma}

\begin{proof}
    Consider the \matching{} realizing $\distFr{\cZ}{ \pnt \pntA}$,
    and break it into three portions:
    \begin{compactitem}
        \item the portion matching $\SC{\cZ}{u_1}{u_i}$ with a
        ``prefix'' $\pnt\pnt' \subseteq \pnt\pntA$,
        \item the portion matching $\SC{\cZ}{u_i}{u_j}$ with a
        subsegment $\pnt'\pntA' \subseteq \pnt\pntA$, and
        \item the portion matching $\SC{\cZ}{u_j}{u_n}$ with a
        ``suffix'' $\pntA'\pntA \subseteq \pnt\pntA$.
    \end{compactitem}
    Now, by \lemref{f:r:basic} \itemref{shortcut}, we have that
    \begin{align*}
        \distFr{\cZ}{ \pnt \pntA}%
        &=%
        \max \pth{\MakeBig%
           \,%
           \distFr{\SC{\cZ}{u_1}{u_i}}{\pnt \pnt'}, %
           \, \,%
           \distFr{\SC{\cZ}{u_i}{u_j}}{\pnt' \pntA'},%
           \, \,%
           \distFr{\SC{\cZ}{u_j}{u_n}}{\pntA' \pntA}%
           \,%
        }%
        \\&
        \geq%
        \max \pth{\MakeBig%
           \, \distFr{\SC{\cZ}{u_1}{u_i}}{\pnt \pnt'},%
           \, \, %
           \distFr{u_i u_j }{\pnt' \pntA'},%
           \, \, %
           \distFr{\SC{\cZ}{u_j}{u_n}}{\pntA' \pntA}%
           \, }%
        \\&%
        \geq %
        \distFr{\cZ'}{\pnt \pntA}.
    \end{align*}
    \aftermathA
\end{proof}

\subsection{Stage 1: Achieving a constant-factor approximation}
\seclab{f:r:query:constant}%
\seclab{queries:stage:one}

In this section we describe a data structure that preprocesses a curve
$\cZ$ to answer queries for the \Frechet distance of a subcurve of
$\cZ$ to a query segment up to a constant approximation factor. This
data structure will be the basis for later extensions.

A query is specified by points $u,v, \pnt$ and $\pntA$. Here $u$ and
$v$ are points on $\cZ$ (and we are also given the edges of $\cZ$
containing these two points), and the points $\pnt$ and $\pntA$ define
the query segment. Our goal is to approximate $\distFr{ \pnt
   \pntA}{\SC{\cZ}{u}{v}}$.

\subsubsection{The data structure}

\paragraph{Preprocessing.}
Build a balanced binary tree $\Tree$ on the \emph{edges} of $\cZ$.
Every internal node $\node$ of $\Tree$ corresponds to a subcurve of
$\cZ$, denoted by $\cNode{\node}$. Let $\segNode{\node}$ denote the
spine of $\cNode{\node}$ (\defref{spine}).  For every node, we
precompute its \Frechet distance of the curve $\cNode{\node}$ to the
segment $\segNode{\node}$. Let $\frNode{\node}$ denote this distance.

\paragraph{Answering a query.}
For the time being, assume that $u$ and $v$ are vertices of $\cZ$.  In
this case, one can compute, in $O( \log n)$ time, $k=O(\log n)$ nodes
$\node_1, \ldots, \node_k$ of $\Tree$, such that $\SC{\cZ}{u}{v} =
\cNode{\node_1} \concatOp \cNode{\node_2} \concatOp \cdots \concatOp
\cNode{\node_k}$. We compute the polygonal curve $\cY =
\segNode{\node_1} \concatOp \cdots \concatOp \segNode{\node_k}$, and
compute its \Frechet distance from the segment $\pnt \pntA$. We denote
this distance by $d = \distFr{\pnt \pntA}{\cY}$. We return
\begin{align*}
    \Delta = d + \max_{i=1}^k \frNode{\node_i}
\end{align*}
as the approximate distance between {$\pnt \pntA$} and the subcurve
$\SC{\cZ}{u}{v}$.

\subsubsection{Analysis}

\begin{lemma}
    \lemlab{easy:shortcut}%
    Given a polygonal curve $\cZ$ with $n$ edges, one can preprocess
    it in $O(n \log^2 n)$ time, such that for any pair $u,v$ of
    vertices of $\cZ$ and a segment $\pnt \pntA$, one can compute, in
    $O( \log n \log \log n)$ time, a $3$-approximation to
    $\distFr{\pnt\pntA} {\SC{\cZ}{u}{v}}$.
\end{lemma}

\begin{proof}
    The construction of the data structure and how to answer a query
    is described above. For the preprocessing time, observe that
    computing the \Frechet distance of a segment to a polygonal curve
    with $k$ segments takes $O(k \log k)$ time
    \cite{ag-cfdbt-95}. Hence, the distance computations in each level
    of the tree $\Tree$ take $O(n \log n)$ time, and $O(n \log^2 n)$
    time overall.
    
    As for the query time, computing $\cY$ takes $O( \log n)$ time,
    and computing its \Frechet distance from $\pnt \pntA$ takes $O(
    \log n \log \log n)$ time \cite{ag-cfdbt-95}.
    
    Finally, observe that the returned distance $\Delta$ is a
    realizable \Frechet distance, as we can take the \matching{}
    between $\pnt \pntA$ and $\cY$, and chain it with the \matching{}
    of every edge of $\cY$ with its corresponding subcurve of
    $\cZ$. Clearly, the resulting \matching{} has width at most
    $\Delta$.
    
    Let $t$ be the index realizing $\max_{i=1}^k
    \frNode{\node_i}$. Then, by repeated application of
    \lemref{shortcut:segment}, we have that $d = \distFr{\pnt
       \pntA}{\cY} \le \distFr{\pnt \pntA}{\cZ}$. Thus,
    \begin{align*}
        \Delta &= d + \max_{i=1}^k \frNode{\node_i}%
        =%
        \distFr{\pnt \pntA}{\cY} + \frNode{\node_t}%
        \leq%
        \distFr{\pnt \pntA}{\cZ} +
        \distFr{\segNode{v_t}}{\cNode{v_t}}%
        \\
        &\leq%
        \distFr{\pnt \pntA}{\cZ} + 2 \min_{\pnt'\pntA' \subseteq \pnt
           \pntA} \distFr{\pnt'\pntA'}{\cNode{v_t}}%
        \leq%
        3\distFr{\pnt \pntA}{\cZ}.
    \end{align*}
    To see the last step, consider the \matching{} realizing
    $\distFr{\pnt \pntA}{\cZ}$, and consider the subsegment $\pnt'
    \pntA'$ of $\pnt\pntA$ that is being matched to $\cNode{v_t}
    \subseteq \cZ$. Clearly, $\distFr{\pnt'\pntA'}{\cNode{v_t}} \leq
    \distFr{\pnt \pntA}{\cZ}$.
\end{proof}

\begin{theorem}%
    \thmlab{data:structure:constant}%
    Given a polygonal curve $\cZ$ with $n$ edges, one can preprocess
    it in $O(n \log^2 n)$ time and using $O(n)$ space, such that,
    given a query specified by
    \begin{compactenum}[\quad(i)]
        \item a pair of points $u$ and $v$ on the curve $\cZ$,
        \item the edges containing these two points, and
        \item a pair of points $\pnt$ and $\pntA$,
    \end{compactenum}
    one can compute, in $O( \log n \log \log n)$ time, a
    $3$-approximation to $\distFr{\pnt \pntA} {\SC{\cZ}{u}{v}}$.
\end{theorem}

\begin{proof}
    This follows by a relatively minor modification of the above
    algorithm and analysis. Indeed, given $u$ and $v$ (and the edges
    containing them), the data structure computes the two vertices
    $u', v'$ that are endpoints of these edges that lie between $u$
    and $v$ on the curve. The data structure then concatenates the
    segments $u{} u'$ and $v'{} v$ to the approximation $\cY$ (here
    $\cY$ is computed for the vertices $u'$ and $v'$).  The remaining
    details are as described above.
\end{proof}

\subsection{Stage 2: A segment query to the entire curve}
\seclab{queries:stage:two}

In this section we describe a data structure that preprocesses a curve
to answer queries for the \Frechet distance of the entire curve to a
query segment up to an approximation factor of $(1+\eps)$. We will use
this data structure as a component in our later extensions.

\subsubsection{The data structure}

We need the following relatively easy construction of an exponential
grid. \figref{exp:sprinkle} illustrates the idea.
The details can be found in \cite{d-raapg-13}. 
\begin{lemma}[\cite{d-raapg-13}]\lemlab{exp:grid}
    Given a point $u \in \Re^d$, a parameter $0<\eps\leq 1$ and an interval
    $[\alpha,\beta] \subseteq \Re$ one can compute in $O\pth{\eps^{-d}
    \log\pth{\beta/\alpha}}$ time and space an exponential grid of points $G(u)$, 
    such that for any point $\pnt \in \Re^d$ with $\distX{\pnt}{u} \in
    [\alpha,\beta]$,
    one can compute in constant time a grid point $\pnt' \in G(u)$ with 
    $\distX{\pnt}{\pnt'} \leq (\eps/2) \distX{\pnt}{u}$.
\end{lemma}

\paragraph{Preprocessing.}
We are given a polygonal curve $\cZ$ in $\Re^d$ with $n$ segments, and
we would like to preprocess it for $(1+\eps)$-approximate \Frechet
distance queries against a query segment. To this end, let $L =
\distFr{uv}{\cZ}$, where $uv$ is the spine of $\cZ$.  We construct an
exponential grid $G(u)$ of points around $u$ with the range 
$[\alpha,\beta]=[\gridMin,L/\eps]$ as described in \lemref{exp:grid}
 and illustrated in \figref{exp:sprinkle}.  
We construct the same grid $G(v)$ around the vertex~$v$.

\begin{figure}\centering
\includegraphics{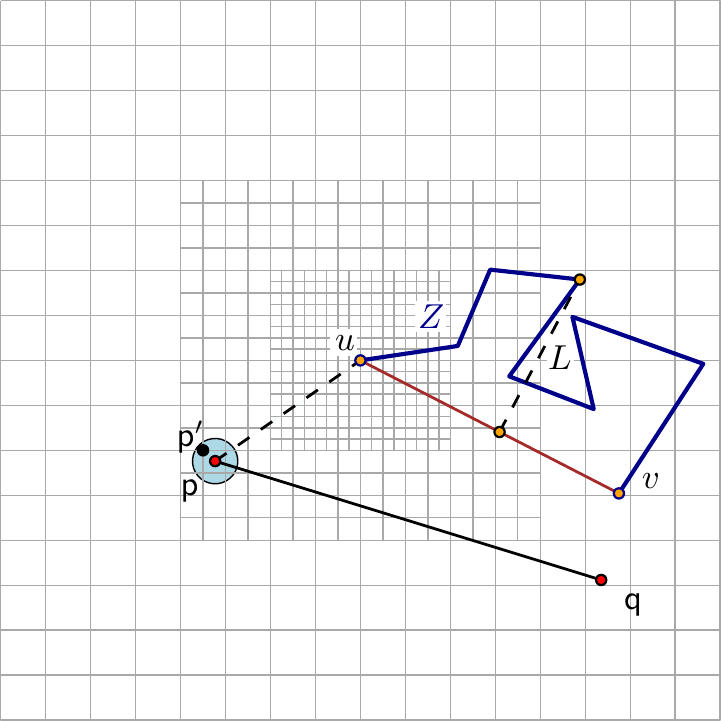}
\caption{{We build an exponential grid around each endpoint of the curve, such that 
for any point $\pnt$, which has distance to the endpoint in the range
$[\gridMin, L/\eps]$, there exists a grid point $\pnt'$ which is relatively close by.}}
\figlab{exp:sprinkle}
\end{figure}

Now, for every pair of points $(\pnt',\pntA') \in G(u) \times G(v)$ we
compute the \Frechet distance $D[\pnt', \pntA'] = \distFr{\pnt'
\pntA'}{\cZ}$ and store it. Thus, we take 
$O\pth{ \GridCompl^2 n \log n }$ time to build a data structure 
that requires $O\pth{ \GridCompl^2 }$ space, where $\GridCompl=\NgridCompl$.

\paragraph{Answering a query.}
Given a query segment $\pnt \pntA$, we compute the distance
\begin{align*}
    r = \max\pth{ \!\MakeSBig \distX{\pnt}{u}, \distX{\pntA}{v}}.
\end{align*}
If $r \leq \gridMin$, then we return $L-r$ as the approximation
to the distance $\distFr{\pnt\pntA}{\cZ}$.  
If $r \geq L/\eps$ then we return $r$ as the approximation. 
Otherwise, let $\pnt'$ (resp., $\pntA'$) be the nearest neighbor to
$\pnt$ in $G(u)$ (resp., $G(v)$). We return the distance
\[
\Delta = D[\pnt', \pntA'] - \max\pth{\distX{\pnt}{\pnt'},
       \distX{\pntA}{\pntA'}}
      \]
as the approximation.

\subsubsection{Analysis}

\begin{lemma}
    \lemlab{epsilon:query:w:curve}%
    Given a polygonal curve $\cZ$ with $n$ vertices in $\Re^d$, one
    can build a data structure, in $O\pth{\GridCompl^2 n \log n )}$ 
    time, that uses $O\pth{ \GridCompl^2 }$ space, such that given a query segment $\pnt \pntA$
    one can $(1+\eps)$-approximate $\distFr{\pnt \pntA}{\cZ}$ in
    $O(1)$ time, where $\GridCompl=\NgridCompl$.
\end{lemma}
\begin{proof}
    The data structure is described above. 
    Given $\pnt \pntA$ we compute the distance of the
    endpoints of this segment from the endpoints of $\cZ$. If they are too
    close, or if one of them is too far away, then we are done since in this
    case the \Frechet distance is dominated either by these distances or by the
    precomputed value $L$.  Otherwise, we find the two cells in the exponential
    grid that contain $\pnt$ and $\pntA$ (that is, the indices of the grid
    points that are close to them) as described above. 
    Using the indices of the grid points, we can directly look-up the
    approximation of the \Frechet distance in constant time.
    
    Now, we argue about the quality of the approximation using the notation
    which is also used above. 
    There are three cases: either
    \begin{inparaenum}[(i)]
    \item $r \leq \gridMin$, or
    \item $L \leq \eps r $, or
    \item $  \gridMin \leq r \leq  L/\eps$.
    \end{inparaenum}
    Let $\Delta$ be the returned value. We claim that in all three cases, it
    holds that
    \begin{align}\Eqlab{sprinkle:claim}
    \Delta \leq \distFr{\pnt\pntA}{\cZ} \leq (1+ \eps) \Delta.
    \end{align}
    First note that by the triangle inequality, 
    \begin{align}\Eqlab{yatq}
    L  - r \leq \distFr{\pnt\pntA}{\cZ} \leq L + r.
    \end{align}
    Now, in case (i) above, $L$ dominates the distance value and we return $\Delta=L-r$.
    Thus, \Eqref{sprinkle:claim} follows from \Eqref{yatq}.

    In case (ii), $r$ dominates the distance value and we return
    $\Delta=r$. Since $r$ is at most $\distFr{\pnt\pntA}{\cZ}$,  again
     \Eqref{sprinkle:claim} follows from \Eqref{yatq}.

    In case (iii), the precomputed \Frechet distance of $\pnt'\pntA'$ to
    $\cZ$ dominates the distance. Recall that we return 
    $\Delta = \distFr{\pnt'\pntA'}{\cZ} - \distFr{\pnt'\pntA'}{\pnt\pntA}$  in
    this case.
    Again, by the triangle inequality, it holds
    that
    \begin{align}\Eqlab{t:eq:2}
         \distFr{\pnt'\pntA'}{\cZ} - \distFr{\pnt'\pntA'}{\pnt\pntA} 
    \leq \distFr{\pnt\pntA}{\cZ} 
    \leq \distFr{\pnt'\pntA'}{\cZ} + \distFr{\pnt'\pntA'}{\pnt\pntA}.
    \end{align}
    Since $r$ is at least $\gridMin$ and by \obsref{f:r:segments}, \lemref{exp:grid} implies that 
    \[\distFr{\pnt'\pntA'}{\pnt\pntA}\leq
    \max\pth{\distX{\pnt}{\pnt'},\distX{\pntA}{\pntA'}} \leq (\eps/2) r,\]
    thus, since also $r$ is at most $\distFr{\pnt\pntA}{\cZ}$ it follows by \Eqref{t:eq:2} that 
    \[
    \Delta 
    \leq \distFr{\pnt\pntA}{\cZ} 
    \leq  \Delta + 2 \distFr{\pnt'\pntA'}{\pnt\pntA} 
    \leq \Delta + \eps r \leq (1+\eps)\Delta.
    \]
    This implies the claim.
\end{proof}

\subsection{Stage 3: A segment query to a subcurve}
\seclab{queries:prelim}%
\seclab{queries:stage:three}%
\seclab{f:r:query}%

In this section we describe a data structure that preprocesses a curve
$\cZ$ to answer queries for the \Frechet distance of a subcurve of
$\cZ$ to a query segment up to an approximation factor of
$(1+\eps)$. For this we combine the data structures developed in the
previous sections.

As in \secref{f:r:query:constant}, a query is defined by two points
$u$ and $v$ on $\cZ$ and a segment with endpoints $\pnt$ and
$\pntA$. The goal is now a $(1+\eps)$-approximation to $\distFr{\pnt
   \pntA}{\SC{\cZ}{u}{v}}$.

\subsubsection{The data structure}

\paragraph{Preprocessing.}
Let $\cZ$ be a given polygonal curve with $n$ vertices. We build the
data structure of \thmref{data:structure:constant}. Next, for each
node of the resulting tree $\Tree$, we build for its subcurve the data
structure of \lemref{epsilon:query:w:curve} using $\eps'=\eps/3$.

\begin{figure}\centering
    \includegraphics{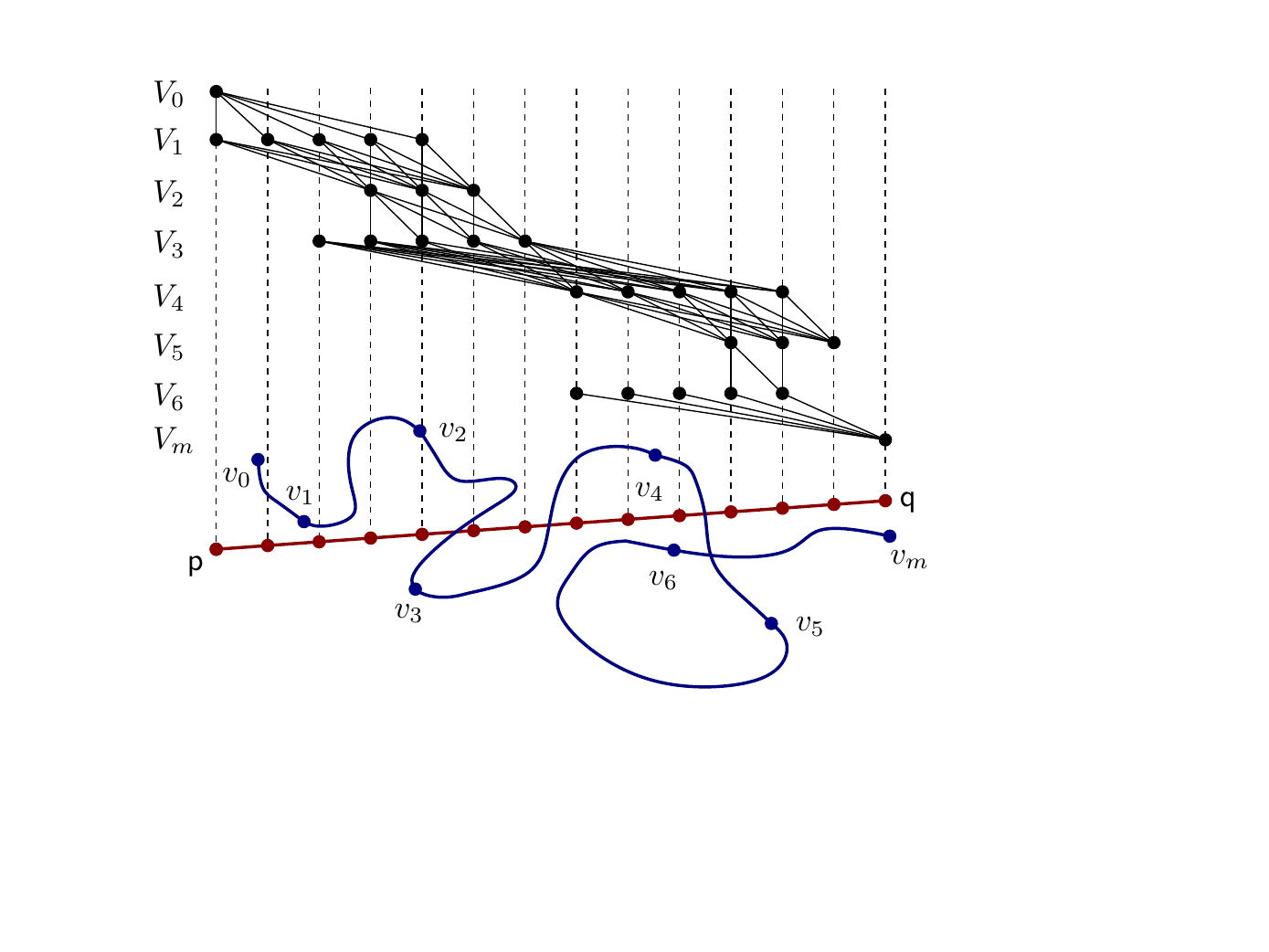}
    \caption{Schematic illustration of the graph $\Graph$ on the
       vertex set $\bigcup V_i$.}
    \figlab{vi:graph}
\end{figure}

\paragraph{Answering a query.}
Using the data structure of \thmref{data:structure:constant} we first
compute a $3$-approximation $r$ to $\distFr{\pnt \pntA}
{\SC{\cZ}{u}{v}}$; that is, $\distFr{\pnt \pntA} {\SC{\cZ}{u}{v}} \leq
r \leq 3\distFr{\pnt \pntA} {\SC{\cZ}{u}{v}}$.  This query also
results in a decomposition of $\SC{\cZ}{u}{v}$ into $m = O( \log n)$
subcurves. Let $u = v_0, v_1, \ldots, v_{m-1}, v_m =v$ be the vertices
of these subcurves, where $v_0v_1$ and $v_{m-1}v_m$ are subsegments of
$\cZ$.

We want to find points on $\pnt \pntA$ that can be matched to
$v_1,\dots,v_{m-1}$ under a $(1+\eps)$-approximate \Frechet matching.
To this end, we uniformly partition the segment $\pnt \pntA$ into
segments of length at most {$\eps r / \constA$, where $\constA$ is a
   sufficiently large constant which we define later.}  Let $\Pi$ be
the set of vertices of this implicit partition.  For each vertex
$v_i$, for $i=1, \ldots, m-1$, we compute its nearest point on $\pnt
\pntA$, and let $V_i \subseteq \Pi$ be the set of all vertices in
$\Pi$ that are in distance at most $2r$ from $v_i$.  The set $V_i$ is
the set of candidate points to match $v_i$ in the \matching{} that
realizes the \Frechet distance.

Now, we build a graph $\Graph$ where $\bigcup_i V_i$ is the multiset
of vertices. Two points $x \in V_i$ and $y \in V_{i+1}$ are connected
by a direct edge in this graph if and only if $y$ is after $x$ in the
oriented segment $\pnt \pntA$. See \figref{vi:graph} for a schematic
illustration. The price of such an edge $\dEdge{x}{y}$ is a
$(1+\eps/4)$-approximation to the \Frechet distance between
$\SC{\cZ}{v_i}{v_{i+1}}$ and $xy$.  The portion
$\SC{\cZ}{v_i}{v_{i+1}}$ of the curve corresponds to a node in
$\Tree$, and this node has an associated data structure that can
answer such queries in constant time (see
\lemref{epsilon:query:w:curve}).  For any point $x \in V_1$, we
directly compute the \Frechet distance $v_0 v_1$ with $\pnt
x$. Similarly, we compute, for each $y \in V_{m-1}$, the \Frechet
distance of the segment $v_{m-1} v_m$ to the segment $y \pntA$. We add
the corresponding edges to $\Graph$ together with the vertices $\pnt$
and $\pntA$.

Using a variant of Dijkstra's algorithm for bottleneck shortest paths,
we now compute a path in this graph which minimizes the maximum cost
of any single edge visited by the path, connecting $\pnt$ with
$\pntA$. The cost of this path is returned as the approximation to the
\Frechet distance between $\SC{\cZ}{u}{v}$ and $\pnt \pntA$.
Intuitively, this path corresponds to the cheapest matching of
$\SC{\cZ}{u}{v}$ (broken into subcurves by the vertices $v_0,\ldots,
v_m$) with $V_0 \times V_1 \times \cdots V_{m-1} \times V_m$, where
$V_0 = \brc{\pnt}$, $V_m = \brc{\pntA}$, and every subcurve
$\SC{\cZ}{v_i}{v_{i+1}}$ is matched with two points in the
corresponding sets $V_i$ and $V_{i+1}$.

\subsubsection{Analysis}

\paragraph{Query time.}
Computing the set of vertices $v_0,v_1, \ldots, v_m$ takes $O( m) = O(
\log n)$ time.  The graph $\Graph$ has $N=O(m/\eps)$ vertices and they
can be computed in $O(m /\eps )$ time. In particular, the number of
vertices in $V_i$ is bounded by $O(1/\eps)$, since they are spread
apart on a line segment by $\eps r/\constA$ and contained inside a
ball of radius $2r$. Thus, the graph has $O\pth{\pth{1/\eps}^{-2}}$
edges connecting $V_i$ with $V_{i+1}$ and $M=O\pth{m/\eps^2}$ edges in
total. The cost of each edge can be computed in constant time, see
\lemref{epsilon:query:w:curve}.  Computing the cheapest path between
$\pnt$ and $\pntA$ in $\Graph$ can be done in $O(N\log N + M)=O\pth{
   (m/\eps) \log (m/\eps) + m/\eps^2}$ time, using Dijkstra's
algorithm for bottleneck shortest paths. Overall, the query time is
\[O\pth{ m + (m/\eps)\log (m/\eps) + m/\eps^2} = O\pth{ \eps^{-2} \log
   n \log\log n}.\]

\begin{figure}\centering
    \includegraphics{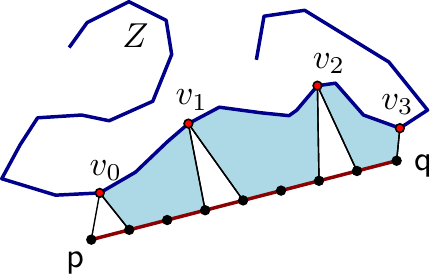}
    \caption{Illustration of the error introduced by snapping. }
    \figlab{snapping-error}
\end{figure}

\paragraph{Quality of approximation.}

Consider the \matching{} that realizes the \Frechet distance between
the query segment $\pnt \pntA$ and the subcurve $\SC{\cZ}{u}{v}$, and
break it at the vertices of $v_0, \ldots, v_m$. Now, snap the matching
such that the endpoints of $\SC{\cZ}{v_i}{v_{i+1}}$ are mapped to
their closest vertices in $V_i$ and $V_{i+1}$, respectively, for all
$i$. This introduces an error of at most $\eps r / \constA \leq
(\eps/3)\distFr{\SC{\cZ}{u}{v}}{\pnt\pntA}$, if we choose $\constA
\geq 9$, see \figref{snapping-error} for an illustration.  We get
another factor of $(1+\eps')=(1+\eps/3)$ error since we are
approximating the price of these portions using
\lemref{epsilon:query:w:curve}. Therefore, the approximation has price
at most
\[
(1+\eps/3)(1+\eps/3)\cdot\distFr{\SC{\cZ}{u}{v}}{\pnt\pntA} \leq
(1+\eps)\cdot\distFr{\SC{\cZ}{u}{v}}{\pnt\pntA}.\]

\paragraph{Preprocessing time and space.}
Building the data structure described in
\thmref{data:structure:constant} takes $O(n \log^2 n)$ time. For each
node $v$ of this tree, building the data structure of
\lemref{epsilon:query:w:curve} takes $O\pth{ \GridCompl^2 \lNode{v}
   \log \lNode{v} }$ time per node, where $\lNode{v}$ is the number of
vertices of the curve stored in the subtree of $v$. As such, overall,
the preprocessing time is $O\pth{\GridCompl^2 n \log^2 n }$. For each
node, this data structure requires $O\pth{ \GridCompl^2 }$ space and
thus the overall space usage is $O\pth{ \GridCompl^2 n }$, where
$\GridCompl = \NgridCompl$.


Putting the above together, we get the following result.

\begin{theorem}%
    \thmlab{segment:queries:f:r}%
    Given a polygonal curve $\cZ$ with $n$ vertices in $\Re^d$, one
    can build a data structure, in $O\pth{\GridCompl^2 n \log^2 n }$
    time, that uses $O\pth{n \GridCompl^2}$ space, such that for a
    query segment $\pnt \pntA$, and any two points $u$ and $v$ on the
    curve (and the segments of the curve that contain them), one can
    $(1+\eps)$-approximate the distance
    $\distFr{\SC{\cZ}{u}{v}}{\pnt\pntA}$ in $O\pth{ \eps^{-2} \log n
       \log\log n}$ time, and $\GridCompl=\NgridCompl$.
\end{theorem}

We emphasize that the result of \thmref{segment:queries:f:r} assumed
nothing on the input curve $\cZ$. In particular, the curve $\cZ$ is
not necessarily $c$-packed.

\section{Universal vertex permutation and %
   its applications}
\seclab{f:r:simpl}

We would like to extend the data structure described in
\secref{queries:stage:one} to support queries with curves of more than
one segment. For this, we first introduce a new method to represent a
polygonal curve in a way such that we can extract a simplification
with a small number of segments quickly. We describe this method in
\secref{u:simpl} and we describe the extension of the data structure
in \secref{k:seg:query}.

\subsection{Universal vertex permutation}
\seclab{u:simpl}

We use the data structure described in \secref{queries:prelim} to
preprocess $\cZ$, such that, given a number of vertices $k \in \Na$,
we can quickly return a simplification of $\cZ$ which has
\begin{compactenum}[(i)]
    \item $\SNumVertices{k}$ vertices of the original curve and
    \item minimal \Frechet distance to $\cZ$, up to a constant factor,
    compared to any simplification of $\cZ$ with only $k$ vertices.
\end{compactenum}
The idea is to compute a permutation of the vertices, such that the
curve formed by the first $k$ vertices in this permutation is a good
approximation to the optimal simplification of a curve using (roughly)
$k$ vertices.

\begin{defn}
    Let $\cZ$ be a polygonal curve with vertices
    $\VertexSet{\cZ}$. Let $\VtxSet \subseteq \VertexSet{\cZ}$ be a
    subset of the vertices that contains the endpoints of $\cZ$.  We
    call the polygonal curve obtained by connecting the vertices in
    $\VtxSet$ in their order along $\cZ$ a \emphi{spine curve} of
    $\cZ$ and we denote it with $\SpCrv{\cZ}{\VtxSet}$. Additionally
    we may call $\SpCrv{\cZ}{\VtxSet}$ a \emphi{k-spine curve} of
    $\cZ$ if it has $k$ vertices.
\end{defn}

\begin{defn}
    Given a polygonal curve $\cZ$ and a permutation $\Phi=
    \permut{\vtxB_1,\dots,\vtxB_n}$ of the vertices of $\cZ$, where
    $\vtxB_{1}$ and $\vtxB_2$ are the endpoints of $\cZ$, let
    $\VtxSet_i$ be the subset $\{\vtxB_j ~|~ 1 \leq j \leq i \}$ of
    the vertices for any $2\leq i\leq n$.  We call $\Phi$ a
    \emphi{universal {vertex} permutation} if it holds that
    \begin{compactenum}[\quad(i)]
        \item $\constA \distFr{\SpCrv{\cZ}{\VtxSet_i}}{\cZ} \geq
        \distFr{\SpCrv{\cZ}{\VtxSet_{i+1}}}{\cZ}$, for any $2 \leq i <
        n$, and
        \item $\distFr{\SpCrv{\cZ}{\VtxSet_{i}}}{\cZ} \leq \constB
        \distFr{\cY}{\cZ}$, for any polygonal curve $\cY$ with
        $\ceil{i/\constC}$ vertices,
    \end{compactenum}
    where $\constA$, $\constB$ and $\constC$ are constants larger than
    one which do not depend on $n$.
\end{defn}

\subsubsection{Construction of the permutation}

We compute a universal vertex permutation of $\cZ$.  The idea of the
algorithm is to estimate for each vertex the error introduced by
removing it, and repeatedly remove the vertex with the lowest error in
a greedy fashion.

Specifically, for each vertex $\vtx$ that is not an endpoint of $\cZ$,
let $\prevX{\vtx}$ be its predecessor on $\cZ$ and let $\nextX{\vtx}$
be its \postdecessor on $\cZ$. Let $\weightCX{\vtx}$ be a
\constN-approximation of
$\distFr{\SC{\cZ}{\prevX{\vtx}}{\nextX{\vtx}}}
{\segX{\prevX{\vtx}}{\nextX{\vtx}}}$.  Insert the vertex $\vtx$ with
weight $\weightCX{\vtx}$ into a min-heap $\heap$.  Repeat this for all
the internal vertices of $\cZ$.

At each step, the algorithm extracts the vertex $\vtx$ from the heap
$\heap$ having minimum weight. Let $\vtxA =
\prevX{\vtx}\pth{\SpCrv{\cZ}{\heap}}$ and
$\vtxC=\nextX{\vtx}\pth{\SpCrv{\cZ}{\heap}}$ be the predecessor and
\postdecessor of $\vtx$ in the curve $\SpCrv{\cZ}{\heap}$,
respectively, where $\heap$ denotes the set of vertices currently in
the heap with the addition of the two endpoints of $\cZ$.

The algorithm removes $\vtx$ from $\heap$ and updates the weight of
$\vtxA$ and $\vtxC$ in $\heap$ (if the vertex being updated is an
endpoint of $\cZ$ its weight is $+ \infty$ and its weight is not being
updated). Updating the weight of a vertex $\vtxA$ is done by computing
its predecessor and \postdecessor vertices in the current curve
$\SpCrv{\cZ}{\heap}$ (i.e., $\prevX{\vtxA} =
\prevX{\vtxA}\pth{\SpCrv{\cZ}{\heap}}$ and $\nextX{\vtxA} =
\nextX{\vtxA}\pth{\SpCrv{\cZ}{\heap}}$) and approximating the \Frechet
distance of the subcurve of (the \emph{original} curve) $\cZ$ between
these two vertices and the segment $\segX{\prevX{\vtxA}}{
   \nextX{\vtxA}}$. Formally, the updated weight of $\vtxA$ is
$\weightCX{\vtxA}$, which is a \constN-approximation to
\begin{align*}
    \distFr{\SC{\cZ}{\prevX{\vtxA}}{\nextX{\vtxA}}}
    {\segX{\prevX{\vtxA}}{\nextX{\vtxA}}}.
\end{align*}
The updated weight of $\vtxC$ is computed in a similar fashion.

The algorithm stops when $\heap$ is empty. Reversing the order of the
handled vertices, results in a permutation
$\permut{\vtx_1,\dots,\vtxB_n}$, where $\vtx_1$ and $\vtxB_2$ are the
two endpoints of $\cZ$.%



\paragraph{Implementation details.}

Using \thmref{segment:queries:f:r}, the initialization takes $O\pth{ n
   \log^2 n}$ time overall, using $\eps=\constEps$.  In addition, the
algorithm keeps the current set of vertices of $\heap$ in a doubly
linked list in the order in which the vertices appear along the
original curve $\cZ$. In each iteration, the algorithm performs one
extract-min from the min-heap $\heap$, and calls the data structure of
\thmref{segment:queries:f:r} twice to update the weight of the two
neighbors of the extracted vertex. As such, overall, the running time
of this algorithm is $O\pth{ n \log^2 n }$.

\paragraph{Extracting a spine curve quickly.}
Given a parameter $K$, we would like to be able to quickly compute the
spine curve $\SpCrv{\cZ}{V_K}$, where $V_K = \brc{ \vtx_1, \ldots,
   \vtx_K}$. To this end, we compute for $i=1, \ldots \floor{\log_2
   n}$, the spine curve $\SpCrv{\cZ}{V_{2^i}}$ by removing the unused
vertices from $\SpCrv{\cZ}{V_{2^{i+1}}}$.  Naturally, we also store
the original curve $\cZ$. Clearly, one can store these $O( \log n)$
curves in $O(n)$ space, and compute them in linear time. Now, given
$K$, one can find the first curve in this collection that has more
vertices than $K$, copy it, and remove from it all the unused
vertices. Clearly, this query can be answered in $O(K)$ time.

\subsubsection{Analysis}

\begin{lemma}%
    \lemlab{spine:curve:f:r}%
    Let $\permut{\vtx_1, \ldots, \vtx_n}$ be the permutation computed
    above. Consider a value $k$, and let $V_k = \brc{\vtxA_1, \ldots,
       \vtxA_k}$ be an ordering of the vertices of ${\vtx_1, \ldots,
       \vtx_k}$ by their order along $\cZ$. Then, it holds that
    $\distFr{\cZ}{\SpCrv{\cZ}{V_k}} \leq \max_{1 \leq i \leq k-1}
    \distFr{\SC{\cZ}{\vtxA_i}{\vtxA_{i+1}}}
    {\segX{\vtxA_i}{\vtxA_{i+1}}}$.
\end{lemma}
\begin{proof}
    This is immediate as one can combine for $i=1, \ldots, k-1$, the
    \matching{}s realizing $\distFr{\SC{\cZ}{\vtxA_i}{\vtxA_{i+1}}}
    {\segX{\vtx_i}{\vtxA_{i+1}}}$ to obtain \matching{}s of
    $\SpCrv{\cZ}{V_k}$ and $\cZ$, and such that the \Frechet distance
    is the maximum used in any of these \matching{}s.
\end{proof}

Let $\vtx_1,\dots,\vtx_n$ be the permutation of the vertices of $\cZ$
as computed in the preprocessing stage, and let $\weight{i}$ denote
weight of vertex $\vtx_i$ at the time of its extraction.  We have the
following three lemmas to prove that the computed permutation is
universal.

\begin{lemma}%
    \lemlab{universal:0}%
    For any $1 \leq i \leq n$, it holds that $\max_{i \leq j \leq n}
    \weight{j} \leq 4 \weight{i}$.
\end{lemma}
\begin{proof}
    We show that the weight of a vertex at the time of extraction is
    at most $4$ times smaller than the final weight of any of the
    vertices extracted before this vertex.  Let $\vtx_i$ be a vertex
    and let $\weightIdx{j}{\vtx_i}$ be the weight of this vertex at
    the time of extraction of some other vertex $\vtx_j$, with $j >
    i$. Clearly, $\weight{j} = \weightIdx{j}{\vtx_j} \leq
    \weightIdx{j}{\vtx_i}$, since the algorithm extracted $\vtx_j$
    with the minimum weight at the time.  If $\weight{i} =
    \weightIdx{i}{\vtx_i} \geq \weightIdx{j}{\vtx_i}$ then the claim
    holds.
    
    Otherwise, if $\weight{i} = \weightIdx{i}{\vtx_i} <
    \weightIdx{j}{\vtx_i}$, then there must be a vertex which caused
    the weight of $\vtx_i$ to be updated. Let $k$ be the minimum index
    such that $j \geq k > i$ and
    $\weightIdx{j}{\vtx_i}=\weightIdx{k}{\vtx_i}$.  We have that
    $\weight{i}$ is a $\const$-approximation of the \Frechet distance
    $\distFr{\segX{\vtxA^i}{\vtxC^i}}{\SC{\cZ}{\vtxA^i}{\vtxC^i}}$ for
    two vertices $\vtxA^i$ and $\vtxC^i$.  Similarly, we have that
    $\weightIdx{k}{\vtx_i}$ is a $\const$-approximation of the
    \Frechet distance
    $\distFr{\segX{\vtxA^k}{\vtxC^k}}{\SC{\cZ}{\vtxA^k}{\vtxC^k}}$ for
    two vertices $\vtxA^k$ and $\vtxC^k$.  Observe that since the
    extraction of $\vtx_k$ caused the weight of $\vtx_i$ to be
    updated, it must be that $\SC{\cZ}{\vtxA^k}{\vtxC^k}$ is a
    subcurve of $\SC{\cZ}{\vtxA^i}{\vtxC^i}$.  Hence, by
    \lemref{monotone:shortcut:base}, we have that
    \begin{align*}
        \constR \cdot \weightIdx{k}{\vtx_i}%
        \leq%
        \distFr{\segX{\vtxA^k}{\vtxC^k}}{\SC{\cZ}{\vtxA^k}{\vtxC^k}}
        \leq%
        3 \distFr{\segX{\vtxA^i}{\vtxC^i}}{\SC{\cZ}{\vtxA^i}{\vtxC^i}}
        \leq%
        3 \cdot \const \cdot \weight{i}.
    \end{align*}
    Now it follows that $\weight{j} \leq \weightIdx{j}{\vtx_i} =
    \weightIdx{k}{\vtx_i} \leq 4 \weight{i}$, which proves the claim.
\end{proof}

\begin{lemma}%
    \lemlab{universal:i}%
    For any $3 \leq i \leq n$ it holds that 
    $\distFr{\SpCrv{\cZ}{\VtxSet_i}}{\cZ} \leq 5 \weight{i+1}$.
\end{lemma}
\begin{proof}
    Let $\vtxA_1,\dots,\vtxA_{i}$ be the vertices in $\VtxSet_i$ in
    the order in which they appear on $\SpCrv{\cZ}{\VtxSet_i}$.
    Consider the mapping between $\cZ$ and this spine curve, which
    associates every edge $\segX{\vtxA_j}{\vtxA_{j+1}}$ of
    $\SpCrv{\cZ}{\VtxSet_i}$ with the subcurve
    $\SC{\cZ}{\vtxA_j}{\vtxA_{j+1}}$. Clearly, it holds that
    \begin{align*}
        \distFr{\cZ}{\SpCrv{\cZ}{\VtxSet_i}} \leq %
        \max_{1 \leq j < i} \distFr{\SC{\cZ}{\vtxA_j}{\vtxA_{j+1}}}
        {\segX{\vtxA_j}{\vtxA_{j+1}}} \leq%
        \const \max_{i < j \leq n} \weight{j}.
    \end{align*}
    Indeed, if $\vtxA_{j+1}$ is the \postdecessor of $\vtxA_j$ on
    $\cZ$, then
    $\distFr{\SC{\cZ}{\vtxA_j}{\vtxA_{j+1}}}{\segX{\vtxA_j}{
          \vtxA_{j+1}}} = 0$, otherwise, there must be a vertex which
    appears on $\cZ$ in between $\vtxA_j$ and $\vtxA_{j+1}$, which is
    contained in $\VtxSet_n \setminus \VtxSet_i$ and the weight of
    this vertex is the approximation of this distance at the time of
    extraction.  Now it follows by \lemref{universal:0} that
    $\distFr{\cZ}{\SpCrv{\cZ}{\VtxSet_i}}\leq 5 \weight{i+1}$.
\end{proof}

\begin{lemma}%
    \lemlab{universal:ii}%
    For any $2 \leq k \leq n/2-1$, let $\cYk$ be the curve with the
    smallest \Frechet distance from $\cZ$ with $k$ vertices (note,
    that $\cYk$ is not restricted to have its vertices lying on
    $\cZ$). We have that $\distFr{\cZ}{\cYk} \geq (5/11) \weight{K+1}
    $, where $K = \SNumVertices{k}$.
\end{lemma}
\begin{proof}
    Let $f:\cYk \rightarrow Z$ be the mapping realizing the \Frechet
    distance between $\cY_k^*$ and $\cZ$. Let $V_i = \permut{ \vtx_1,
       \ldots, \vtx_i}$, for $i=1, \ldots, n$.
    
    \parpic[r]{\includegraphics[scale=1]{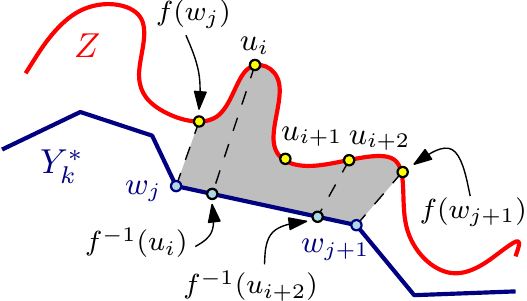}}
    
    Since $\cYk$ has only $k$ vertices, it breaks $\cZ$ into $k-1$
    subcurves. Since, $K \geq 2(k-1) + 1$, there must be three
    consecutive vertices $\vtxA_i, \vtxA_{i+1}, \vtxA_{i+2}$ on
    $\SpCrv{\cZ}{\VtxSet_{K}}$ and two vertices $\vtxC_j,\vtxC_{j+1}$
    of $\cYk$, such that the vertices $\vtxA_i, \vtxA_{i+1},
    \vtxA_{i+2}$ appear on the subcurve $\cZ' =
    \SC{\cZ}{f(\vtxC_j)}{f(\vtxC_{j+1})}$, see the figure on the
    right.
    
    Now, $\segX{f^{-1}(\vtxA_i)}{f^{-1}(\vtxA_{i+2})} \subseteq
    \vtxC_j\vtxC_{j+1}$ and by \lemref{f:r:basic}, we have
    \begin{align*}
        \distFr{\cZ}{\cYk}%
        &\geq %
	\distFr{ \MakeBig \SC{\cZ}{\MakeSBig
              f(\vtxC_j)}{f(\vtxC_{j+1})} }{\vtxC_j \vtxC_{j+1}}%
        \geq %
        \distFr{\SC{\cZ}{\vtxA_i}{\vtxA_{i+2}}} {f^{-1}(\vtxA_i)
           f^{-1}(\vtxA_{i+2})}%
        \\ &%
        \geq %
        \distFr{\SC{\cZ}{\vtxA_i}{\vtxA_{i+2}}} {f^{-1}(\vtxA_i)
           f^{-1}(\vtxA_{i+2})}%
        \geq %
        \frac{1}{2}%
        \distFr{\MakeBig\SC{\cZ}{\vtxA_i}{\vtxA_{i+2}}}%
        {\spineX{\SC{\cZ}{\vtxA_i}{\vtxA_{i+2}}}}%
        \\ &%
        \geq %
        \frac{1}{2} \cdot \frac{10}{11} \weightIdx{K+1}{ \vtxA_{i+1}}%
        \geq %
        \frac{5}{11} \weightIdx{K+1}{\vtx_{K+1}}%
        = %
        \frac{5}{11} \weight{K+1},
    \end{align*}
    as the simplification algorithm removed the minimum weight vertex
    at time $K+1$ (i.e., $\vtx_{K+1}$).
\end{proof}

\subsubsection{The result}

\begin{theorem}
    \thmlab{f:r:simpl:result}%
    Given a polygonal curve $\cZ$ with $n$ edges, we can preprocess it
    using $O(n)$ space and $O\pth{ n\log^2 n }$ time, such that, given
    a parameter $k \in \Na$, we can output in $O(k)$ time a
    $(\SNumVertices{k})$-spine curve ${\cZ}'$ of $\cZ$ and a value
    $\delta$, such that \smallskip
    \begin{compactenum}[\quad(i)]
        \item $\delta/11 \leq \distFr{\cYk}{\cZ}$,
        and \smallskip
        \item $\distFr{{\cZ}'}{\cZ} \leq \delta$,
    \end{compactenum}
    \smallskip%
    where $\cYk$ is the polygonal curve with $k$ vertices with minimal
    \Frechet distance from~$\cZ$.  (For $k \geq n/2$ we output $\cZ$
    and $\delta=0$).
\end{theorem}

\begin{proof}
    The algorithm computing the universal vertex permutation and its
    associated data structure is described above, for $K
    =\SNumVertices{k}$.  Specifically, it returns the spine curve
    $\cZ' = \SpCrv{\cZ}{V_K}$ as the required approximation, with the
    value $\delta=5\weight{K+1}$. Computing $\cZ'$ takes $O(k)$ time.
    By \lemref{universal:i} and \lemref{universal:ii}, we have that
    $\cZ'$ and $\delta$ satisfy the claim.
    
    Building the data structure takes $O\pth{n \log^2 n }$ time, and
    it uses $O(n)$ space using $\eps=\constEps$. Each query to this
    data structure takes $O\pth{\log n \log\log n}$ time. We perform a
    constant number of these queries to the data structure per
    extraction from the heap, thus getting the claimed preprocessing
    time.
\end{proof}

\subsection{Extending the data structure for %
   \Frechet-distance queries}
\seclab{k:seg:query}

We use the universal vertex permutation described in the previous
section to extend our data structure of \secref{queries:stage:one} to
support queries with more than one segment.

\subsubsection{The data structure}

The input is a polygonal curve $\cZ \in \Re^d$ with $n$ vertices.

\paragraph{Preprocessing.}
Similar to the algorithm of \secref{f:r:query:constant}, build a
balanced binary tree $\Tree$ on $\cZ$. For every internal node $\node$
of $\Tree$ construct the data structure of \thmref{f:r:simpl:result}
for $\cNode{\node}$, denoted by $\DSX{\node}$, and store it at
$\node$.

\paragraph{Answering a query.}
Given any two vertices $u$ and $v$ of $\cZ$, and a query polygonal
curve $\cQ$ with $k$ segments, the task is to approximate
$\distFr{\cQ}{\SC{\cZ}{u}{v}}$. We initially proceed as in
\secref{f:r:query:constant}, computing in $O( \log n)$ time, $m=O(\log
n)$ nodes $\node_1, \ldots, \node_m$ of $\Tree$, such that
$\SC{\cZ}{u}{v} = \cNode{\node_1} \concatOp \cNode{\node_2} \concatOp
\cdots \concatOp \cNode{\node_k}$.  Now, extract a simplified curve
with $K$ vertices from $\DSX{\node_i}$, denoted by
$\kNode{K}{\node_i}$, for $i=1,\ldots, m$, where $K =
\SNumVertices{k}$.  For $i=1,\ldots, m$, let $\delta_i$ denote the
simplification error (as returned by $\DSX{\node_i}$), where
$\distFr{\kNode{K}{\node_i}}{ \cNode{\node_i}} \leq \delta_i$ and
$\delta_i/11$ is a lower bound to the \Frechet distance of \emph{any}
curve with at most $k$ vertices from $\cNode{\node_i}$, for
$i=1,\ldots, m$ (see \thmref{f:r:simpl:result}).

Next, compute the polygonal curve $\cS = \kNode{K}{\node_1} \concatOp
\cdots \concatOp \kNode{K}{\node_m}$, and its \Frechet distance from
$\cQ$; that is, $d = \distFr{\cS}{\cQ}$.  We return
\begin{align}
    \Delta = d + \max_{1\leq i}^{m}\delta_i,%
    \Eqlab{delta}
\end{align}
as the approximate distance between $\cQ$ and $\SC{\cZ}{u}{v}$.

\subsubsection{Analysis}

\paragraph{Query time.}
Extracting the $m=O(\log n)$ relevant nodes takes $O(\log n)$ time.
Querying these $m$ data structures for the simplification of the
respective subcurves, takes $O(k m)$ overall, by
\thmref{f:r:simpl:result}.  Computing the \Frechet distance between
the resulting simplification $\cS$ of $\SC{\cZ}{u}{v}$, which has $O(m
k)$ edges, and $\cQ$ takes $O(k^2m \log(k^2m))$ time
\cite{ag-cfdbt-95}.  Thus the overall time used for answering a query
is bounded by $O\pth{ m + km + k^2m \log(k^2m)} = O\pth{ k^2m
   \log(km)} = O\pth{ k^2 \log n \log (k \log n ) }$.

\paragraph{Preprocessing time and space.}
Building the initial tree $\Tree$ takes $O(n)$ time and it requires
$O(n)$ space.  Let $\lNode{\node}$ denote the number of vertices of
$\cNode{\node}$.  For each node~$\node$, computing the additional
information and storing it requires $O(\lNode{\node})$ space and
$O\pth{ \lNode{\node} \log^2 \lNode{\node} }$ time.  Recall that
$\Tree$ is a balanced binary tree and for the nodes $\node_1, \ldots,
\node_{t}$ contained in one level of the tree it holds that $\sum_{1
   \leq i}^{t} \lNode{\node_i} = n$. Thus, computing and storing the
additional information takes an additional $O\pth{ n \log^3 n }$ time
and $O(n \log n)$ space by \thmref{f:r:simpl:result}.

\paragraph{Quality of approximation.}
By the following lemma the data structure achieves a constant-factor
approximation.

\begin{lemma}%
    \lemlab{k:seg:query:vertices}%
    Given a polygonal curve $\cZ$ and a query curve $\cQ$ with $k$
    segments, the value $\Delta$ (see \Eqref{delta}) returned by the
    above data structure is a constant-factor approximation to
    $\distFr{\cQ}{\SC{\cZ}{u}{v}}$.
\end{lemma}
\begin{proof}
    Clearly, $\Delta$ bounds the required distance from above, as one
    can extract a \matching{} of $\cQ$ and $\SC{\cZ}{u}{v}$ realizing
    $\Delta$. As such, we need to prove that $\Delta= O(r)$, where $r
    = \distFr{\cQ}{\SC{\cZ}{u}{v}}$.
    
    So, let $f:\cQ \rightarrow \SC{\cZ}{u}{v}$ be the mapping
    realizing $r = \distFr{\cQ}{\SC{\cZ}{u}{v}}$, and let $\cQ_i =
    f^{-1}( \cNode{\node_i})$, for $i=1, \ldots, m$. Clearly, $r =
    \max_i \distFr{\cQ_i}{\cNode{\node_i}}$.  Since $\cQ_i$ has at
    most $k$ vertices, by \thmref{f:r:simpl:result}, we have
    \begin{align}
        \frac{\delta_i}{11} \leq \distFr{\cQ_i}{\cNode{\node_i}} \leq
        r,
        \qquad%
        \text{and}%
        \qquad%
        \distFr{\kNode{K}{\node_i}}{\cNode{\node_i}}%
        \leq%
        \delta_i,%
        \Eqlab{alpha:i:r}
    \end{align}
    for $i=1,\ldots, m$.  In particular, we have $\delta_i \leq 11r$.
    Now, by the triangle inequality, we have that
    \begin{align*}
        \distFr{\kNode{K}{\node_i}}{\cQ_i} %
        \leq%
        \distFr{\kNode{K}{\node_i}}{\cNode{\node_i}} +
        \distFr{\cNode{\node_i}}{\cQ_i}%
        \leq%
        \delta_i + r%
        \leq%
        12r.
    \end{align*}
    As such, $d = \distFr{\cS}{\cQ} \leq \max_i
    \distFr{\kNode{K}{\node_i}}{\cQ_i} \leq 12r$.  Now, $\Delta = d +
    \max_{i}\delta_i \leq 12r + 11r = 23r$.
\end{proof}

\paragraph{The result.}
Putting the above together, we get the following result. We emphasize
that $k$ is being specified together with the query curve, and the
data structure works for any value of $k$.

\begin{theorem}
    \thmlab{k:seg:query:result}%
    Given a polygonal curve $\cZ$ with $n$ edges, we can preprocess it
    in $O(n \log^3 n )$ time and $O(n \log n)$ space, such that, given
    a query specified by
    \begin{compactenum}[\quad(i)]
        \item a pair of points $u$ and $v$ on the curve $\cZ$,
        \item the edges containing these two points, and
        \item a query curve $\cQ$ with $k$ segments,
    \end{compactenum}
    one can approximate $\distFr{\cQ}{\SC{\cZ}{u}{v}}$ up to a
    constant factor in $O\pth{ k^2 \log n \log (k \log n ) }$ time.
\end{theorem}

\begin{proof}
    The preprocessing is described and analyzed above. The query
    procedure needs to be modified slightly since the $u$ and $v$ are
    not necessarily vertices of $\cZ$. However, this can be done the
    same way as for the initial data structure in
    \thmref{data:structure:constant}.  Let $u',v'$ be the first and
    last vertices of $\cZ$ contained in $\SC{\cZ}{u}{v}$. We now
    extract the $m=O(\log n)$ nodes $\node_1,\dots,\node_m$ of
    $\Tree$, such that
    \[\cX=\segX{u}{u'} \concatOp \cNode{\node_1}
    \concatOp \dots \concatOp \cNode{\node_m} \concatOp \segX{v'}{v} =
    \SC{\cZ}{u}{v}.\] We continue with the procedure as described
    above using this node set.  The analysis of
    \lemref{k:seg:query:vertices} applies with minor modifications.
\end{proof}


\section{Conclusions}
\seclab{conclusions}

In this paper, we presented algorithms for approximating the \Frechet
distance when one is allowed to perform shortcuts on the original
curves. More specifically the presented algorithms approximate the
\asymmetric{} \vrestricted{} shortcut \Frechet distance.
Surprisingly, for $c$-packed curves it is possible to compute a
constant factor approximation in a running time which is near linear
in the complexity of the input curves.

We also presented a way to compute an ordering of the vertices of the
curve, such that any prefix of this ordering serves as a good
approximation to the curve in the \Frechet distance, and it is optimal
up to constant factors.
We used this universal vertex permutation to develop a data
structure that can quickly approximate (up to a constant factor) the
(regular) \Frechet distance between a query curve and the input
curve. Surprisingly, the query time is logarithmic in the complexity
of the original curve (and near quadratic in the complexity of the
query curve).

There are many open questions for further research. The most immediate
questions being how to extend our result to the other definitions of a
shortcut \Frechet distance mentioned in the introduction and how to
improve the approximation factor. The work in this paper is a step
towards solving these more difficult questions.  

As for exact computations, it is easy to see that one can obtain
polynomial-time algorithms by modifying the algorithms presented in this paper
even for general polygonal curves, see also \cite{d-raapg-13}.
Surprisingly, a more recent result shows that if the requirement that
shortcuts have to start and end at input vertices is dropped, the
problem of computing the shortcut \Frechet distance becomes
\NPHard~\cite{bds-jnp-13, d-raapg-13}.

\paragraph*{Acknowledgments.}

The authors thank Mark d{}e Berg, Marc van Kreveld, Benjamin Raichel,
Jessica Sherette, and Carola Wenk for insightful discussions on the
problems studied in this paper and related problems.  The authors also
thank the anonymous referees for their detailed and insightful
comments.

\bibliographystyle{alpha}%
\bibliography{shortcut}%

\end{document}

%% file: prefix.tex
\usepackage{tikz}  
\usepackage{amsmath}%
\usepackage{lmodern}%

\usepackage{hyperref}%
\hypersetup{%
   breaklinks,%
   ocgcolorlinks,
   colorlinks=true,%
   linkcolor=[rgb]{0.45,0.0,0.0},%
   citecolor=[rgb]{0,0,0.45}
}%

\usepackage[hyperref, amsmath,thmmarks]{ntheorem}
\theoremseparator{.}

\usepackage{titlesec}
\titlelabel{\thetitle. }

\usepackage{graphicx}
\usepackage{ifpdf}
\usepackage{xspace}
\usepackage{paralist}
\usepackage{picins}
\usepackage[noend]{algorithmic}
\usepackage{xspace}
\usepackage{color}

\usepackage{euscript}
\usepackage{amssymb}

\usepackage{wide}

\DefineNamedColor{named}{RedViolet} {cmyk}{0.07,0.90,0,0.34}

\newcommand{\AlgorithmI}[1]{{\textcolor[named]{RedViolet}{\texttt{\bf{#1}}}}}
\newcommand{\Algorithm}[1]{{\AlgorithmI{#1}\index{algorithm!#1@{\AlgorithmI{#1}}}}}

\definecolor{darkblue}{rgb}{0,0,0.35}

\DefineNamedColor{named}{OliveGreen}{cmyk}{0.64, 0, 0.95, 0.40}
\providecommand{\ComplexityClass}[1]{{{\textcolor[named]{OliveGreen}{%
      \textsc{#1}}}}}
\providecommand{\NPHard}{{\ComplexityClass{NP-Hard}}%
   \index{NP!hard}\xspace}

\newtheorem{theorem}{Theorem}[section]

\newtheorem{lemma}[theorem]{Lemma}

\newtheorem{datastructure}[theorem]{Data-Structure}

\newtheorem{observation}[theorem]{Observation}

\theorembodyfont{\rm}
     \theoremheaderfont{\normalfont\bfseries}
\newtheorem{remark}[theorem]{Remark}%
\newtheorem{defn}[theorem]{Definition}
\theoremheaderfont{\em}%
\theorembodyfont{\upshape}%
\theoremstyle{nonumberplain}%
\theoremseparator{}%
\theoremsymbol{\rule{2mm}{2mm}}%
\newtheorem{proof}{{P}roof:}%

\newcommand{\atgen}{\symbol{'100}}
\newcommand{\SarielThanks}[1]{\thanks{Department of Computer
      Science; 
      University of Illinois; 
      201 N. Goodwin Avenue;
      Urbana, IL, 61801, USA;
      {\tt sariel\atgen{}uiuc.edu}; {\tt
         {{http://www.uiuc.edu/\string~sariel/}}.} #1}}

\providecommand{\si}[1]{#1}

\newcommand{\seclab}[1]{\label{sec:#1}}
\newcommand{\secref}[1]{Section~\ref{sec:#1}}
\newcommand{\secrefpage}[1]{Section~\ref{sec:#1}%
           $_\text{p\pageref{sec:#1}}$}

\newcommand{\thmlab}[1]{{\label{theo:#1}}}
\newcommand{\thmref}[1]{Theorem~\ref{theo:#1}}
\newcommand{\thmrefpage}[1]{Theorem~\ref{theo:#1}
           $_\text{p\pageref{theo:#1}}$}

\newcommand{\remlab}[1]{\label{rem:#1}}

\newcommand{\remrefpage}[1]{Remark~\ref{rem:#1}%
   $_\text{p\pageref{rem:#1}}$}

\newcommand{\lemlab}[1]{\label{lemma:#1}}
\newcommand{\lemref}[1]{Lemma~\ref{lemma:#1}}
\newcommand{\lemrefpage}[1]{Lemma~\ref{lemma:#1}%
   $_\text{p\pageref{lemma:#1}}$}

\newcommand{\figlab}[1]{\label{fig:#1}}
\newcommand{\figref}[1]{Figure~\ref{fig:#1}}

\newcommand{\deflab}[1]{\label{def:#1}}
\newcommand{\defref}[1]{Definition~\ref{def:#1}}

\newcommand{\eqlab}[1]{\label{equation:#1}}
\newcommand{\Eqref}[1]{Eq.~(\ref{equation:#1})}
\newcommand{\Eqrefpage}[1]{Eq.~(\ref{equation:#1})%
   $_\text{p\pageref{equation:#1}}$}

\newcommand{\dsref}[1]{Data-Structure~\ref{ds:#1}}
\newcommand{\dslab}[1]{\label{ds:#1}}

\newcommand{\obslab}[1]{\label{observation:#1}}
\newcommand{\obsref}[1]{Observation~\ref{observation:#1}}

\newcommand{\itemlab}[1]{\label{item:#1}}
\newcommand{\itemref}[1]{(\ref{item:#1})}

\newcommand{\cardin}[1]{\left| {#1} \right|}

\newcommand{\aftermathA}{\par\vspace{-\baselineskip}}

\newcommand{\pbrc}[2][\!\!]{#1\left[ {#2} \MakeBig \right]}

\newcommand{\ts}{\hspace{0.6pt}}

\newcommand{\IncGraphPage}[4][]{ \ifpdf
   \IncludeGraphics[page=#4,#1]{\File{#2/#3}}
   \else
   \IncludeGraphics[#1]{\File{#2/#3/#3_#4}}
   \fi } \newcommand{\IncludeGraphics}[2][]{%
   \includegraphics[#1]{#2}
}
\newcommand{\File}[1]{#1}
\newcommand{\etal}{\textit{et~al.}\xspace}

\definecolor{blue25}{rgb}{0,0,0.55}
\newcommand{\emphic}[2]{%
     \textcolor{blue25}{%
         \textbf{\emph{#1}}}%
         \index{#2}}

\newcommand{\emphi}[1]{\emphic{#1}{#1}}

\newcommand{\Frechet}{Fr\'{e}c{h}e{}t\xspace}%

\newcommand{\widthX}[2]{{\mathrm{width}}_{#1}\pth{#2}}

\newcommand{\distFr}[2]{\mathsf{d}_{\EuScript{F}}\pth{#1, #2}}
\newcommand{\distSFr}[3]{\mathsf{d}_{\EuScript{S}}\pth{#1,#2, #3}}
\newcommand{\distoSFr}[2]{\mathsf{d}_{\EuScript{S}}\pth{#1,#2}}

\newcommand{\distCmd}[1]{\left\| {#1} \right\|}
\newcommand{\distX}[2]{\distCmd{#1 - #2}}

\newcommand{\Reached}{R}
\newcommand{\Queue}{Q}

\newcommand{\pnt}{\mathsf{p}}
\newcommand{\pntA}{\mathsf{q}}
\newcommand{\pntB}{\mathsf{r}}

\newcommand{\pntC}{\mathsf{u}}

\newcommand{\pntD}{\mathsf{v}}

\newcommand{\pntF}{\mathsf{f}}
\newcommand{\pntG}{\mathsf{g}}

\newcommand{\x}{x}
\newcommand{\y}{y}

\newcommand{\ACoord}[1]{\x_{#1}}
\newcommand{\BCoord}[1]{\y_{#1}}

\newcommand{\xPnt}{\ACoord{\pnt}}
\newcommand{\yPnt}{\BCoord{\pnt}}

\newcommand{\xPntA}{\ACoord{\pntA}}
\newcommand{\yPntA}{\BCoord{\pntA}}

\newcommand{\xPntB}{\ACoord{\pntB}}
\newcommand{\yPntB}{\BCoord{\pntB}}

\newcommand{\opt}{\mathrm{opt}}

\newcommand{\IEdge}[2]{\mathcal{I}_{\mathrm{edge}}\pth{#1, #2}}

\newcommand{\SimplifyX}[1]{#1}

\newcommand{\SubCurvifyY}[1]{{#1}}

\newcommand{\cQ}{\mathsf{Q}}
\newcommand{\cS}{\mathsf{S}}
\newcommand{\cZ}{{Z}}

\newcommand{\cXBase}{X} 
\newcommand{\cYBase}{Y} 

\newcommand{\cX}{\SimplifyX{\cXBase}}
\newcommand{\cXOrig}{\pmb{\cXBase}}

\newcommand{\subcX}[1]{\SubCurvifyY{\cX}_{#1}}

\newcommand{\cY}{\SimplifyX{\cYBase}}
\newcommand{\cYOrig}{\pmb{\cYBase}}

\newcommand{\subsegY}[1]{\overline{\cY}_{#1}}

\newcommand{\linelab}[1]{\label{line:#1}}
\newcommand{\lineref}[1]{Line~\ref{line:#1}}

\newcommand{\seg}{\mathsf{h}}
\newcommand{\segA}{\mathsf{u}}

\newcommand{\FDleqI}[2]{\ensuremath{\EuScript{R}^{#1}_{\leq #2}}}

\newcommand{\ShowComments}[1]{}

\newcommand{\FullFDleqC}{\EuScript{D}}
\newcommand{\FullFDleq}[1]{\FullFDleqC_{\leq #1}}
\newcommand{\NleqC}{\EuScript{N}}
\newcommand{\Nleq}[1]{\NleqC_{\leq #1}}

\newcommand{\Graph}{\mathsf{G}}

\providecommand{\pth}[2][\!]{#1\left({#2}\right)}
\providecommand{\brc}[1]{\left\{ {#1} \right\}}
\providecommand{\sep}[1]{\,\left|\, {#1} \MakeBig\right.}
\providecommand{\MakeBig}{\rule[-.2cm]{0cm}{0.4cm}}

\newcommand{\constA}{c_1}
\newcommand{\constB}{c_2}
\newcommand{\constC}{c_3}

\newcommand{\PntSet}{{\mathsf{P}}}

\newcommand{\sRadius}{\mu}

\newcommand{\CellXY}[2]{C_{#1,#2}}

\providecommand{\MakeBig}{\rule[-.2cm]{0cm}{0.4cm}}
\providecommand{\MakeSBig}{\rule[0.0cm]{0.0cm}{0.35cm}} 

\newcommand{\IVCellXY}[2]{I^v_{#1,#2}}
\newcommand{\IHCellXY}[2]{I^h_{#1,#2}}
\newcommand{\RVCellXY}[2]{R^v_{#1,#2}}
\newcommand{\RHCellXY}[2]{R^h_{#1,#2}}

\providecommand{\eps}{{\varepsilon}}%
\renewcommand{\Re}{{\rm I\!\hspace{-0.025em} R}}

\newcommand{\NULL}{\ensuremath{\mathsf{null}}}

\providecommand{\ceil}[1]{\left\lceil {#1} \right\rceil}
\providecommand{\floor}[1]{\left\lfloor {#1} \right\rfloor}

\newcommand{\density}{\phi}
\renewcommand{\th}{t{}h\xspace}

\newcommand{\concatOp}{\oplus}
\newcommand{\pbrcZ}[1]{\left[ \MakeSBig {#1} \right]}

\newcommand{\SC}[3]{{#1}\!\!\ts\ts%
   \left\langle  #2, #3 \right\rangle}

\newcommand{\SubCrv}[3]{{#1}\!\!\ts\ts\pbrcZ{#2, #3}}
\newcommand{\ScutCrv}[3]{\overline{#1}\!\pbrcZ{#2, #3}}
\newcommand{\SpCrv}[2]{{#1}_{#2}}

\newcommand{\Tree}{T}
\newcommand{\node}{\nu}

\newcommand{\VertexSet}[1]{{V}\pth{#1}}
\newcommand{\VtxSet}{{V}}

\newcommand{\cNode}[1]{\mathrm{c{}r}\pth{#1}}

\newcommand{\kNode}[2]{\mathrm{s{}i{}m{}p{}l_{#1}}\pth{#2}}
\newcommand{\segNode}[1]{\mathrm{s{e}g}\pth{#1}}
\newcommand{\frNode}[1]{\mathrm{d}_{#1}}

\newcommand{\lNode}[1]{\mathrm{l}\pth{#1}}

\newcommand{\dEdge}[2]{{#1 \rightarrow #2}}

\newcommand{\Na}{{\rm I\!\hspace{-0.025em} N}}
\newcommand{\simpX}[1]{\mathrm{s{i}m{p}l}\pth{#1}}
\newcommand{\spineX}[1]{\mathrm{s{}p{}i{}n{}e}\pth{#1}}

\DefineNamedColor{named}{RedViolet}{cmyk}{0.07,0.90,0,0.34}

\newcommand{\tunnelTest}{\Algorithm{tunnel}\xspace}

\newcommand{\deciderFr}{\Algorithm{decider}\xspace}
\newcommand{\DeciderFr}{\Algorithm{Decider}\xspace}

\newcommand{\edge}{\mathsf{e}}

\newcommand{\frVal}[1]{\mathsf{d}\pth{#1}}

\newcommand{\constSC}{3}
\newcommand{\constDec}{4}

\newcommand{\tunnel}{tunnel\xspace}

\newcommand{\tunnels}{tunnels\xspace}
\newcommand{\Tunnels}{Tunnels\xspace}

\newcommand{\tunnelLtr}{\mathsf{\tau}}
\newcommand{\xtunnel}[2]{\tunnelLtr\pth{ #1,  #2}}
\newcommand{\xTunnels}[4]{\mathcal{T}\pth{#1, #2, #3, #4}}
\newcommand{\xTunnelsLeq}[5]{\mathcal{T}_{\leq #5}\pth{#1, #2, #3, #4}}

\newcommand{\scPrice}[2]{{\mathrm{p{r}c}}\pth{\xtunnel{#1}{#2}}}

\newcommand{\rCreate}[4]{\mathrm{r}_{\mathrm{crt}}\pth{#1, #2, #3,#4}}
\newcommand{\rMinCreate}[4]{\mathrm{r}_{\min}\pth{#1, #2, #3,#4}}
\newcommand{\bCanonical}[4]{\mathrm{\tau}_{\min}\pth{#1, #2, #3,#4}}

\newcommand{\priceX}[1]{\mathrm{p{r}c}\pth{#1}}

\newcommand{\Interval}{\mathcal{I}}
\newcommand{\distSet}[2]{\mathsf{d}\pth{#1, #2}}

\newcommand{\AlphaBetaCovered}{$(\alpha,\beta)$-covered\xspace}

\newcommand{\wu}{\widehat{u}}
\newcommand{\wv}{\widehat{v}}

\newcommand{\permut}[1]{\left\langle {#1} \right\rangle}

\newcommand{\vtxA}{u}
\newcommand{\vtxB}{v}
\newcommand{\vtxC}{w}

\newcommand{\FigWidth}{0.9\textwidth}

\newcommand{\RTBoundedK}{O\pth{ c^2 k n \log^3 n }}%
\newcommand{\RTUnbounded}{O\pth{ c^2 n \log^2 n \pth{\log n
                                 +\eps^{-2d}\log(1/\eps)} }}%

\newcommand{\const}{\ensuremath{\frac{11}{10}}}
\newcommand{\constN}{\ensuremath{({11}/{10})}}

\newcommand{\constR}{\ensuremath{\frac{10}{11}}}
\newcommand{\constEps}{\ensuremath{1/10}}
\newcommand{\weightIdx}[2]{\ensuremath{\phi_{#1}\pth{#2}}}
\newcommand{\weight}[1]{\ensuremath{\phi\pth{\vtx_{#1}}}}

\newcommand{\weightCX}[1]{\ensuremath{\phi}_{#1}}
\newcommand{\vtx}{\vtxB}

\usepackage{pifont}%
\usepackage[symbol*]{footmisc}

\newcommand*\circled[1]{\footnotesize\tikz[baseline=(char.base)]{
    \node[shape=circle,draw,inner sep=0.2pt] (char) {#1};}}

\DefineFNsymbols*{stars}[text]{%
   {\ding{172}} %
   {\ding{173}} %
   {\ding{174}} %
   {\ding{175}} %
   {\ding{176}} %
   {\ding{177}} %
   {\ding{178}} %
   {\ding{179}} %
   {\ding{180}} %
   {{\protect\circled{{\tiny 10}}}}%
   {{\protect\circled{{\tiny 11}}}}%
   {{\protect\circled{{\tiny 12}}}}%
   {{\protect\circled{{\tiny 13}}}}%
   {{\protect\circled{{\tiny 14}}}}%
   {{\protect\circled{{\tiny 15}}}}%
   {{\protect\circled{{\tiny 16}}}}%
   {{\protect\circled{{\tiny 17}}}}%
   {{\protect\circled{{\tiny 18}}}}%
   {{\protect\circled{{\tiny 19}}}}%
   {{\protect\circled{{\tiny 20}}}}%
   {{\protect\circled{{\tiny 21}}}}%
   {{\protect\circled{{\tiny 22}}}}%
}%

\newcommand{\segX}[2]{{#1}{#2}}

\newcommand{\DSX}[1]{\mathcal{D}_{#1}}

\newcommand{\prevX}[1]{{#1}^{-}}
\newcommand{\nextX}[1]{{#1}^{+}}

\newcommand{\heap}{\mathcal{H}}

\newcommand{\SNumVertices}[1]{2#1 - 1}

\newcommand{\cYk}{\cY_k^*}

\numberwithin{figure}{section}
\numberwithin{equation}{section}

\newcommand{\postdecessor}{successor\xspace}

\newcommand{\GridCompl}{\ensuremath{\chi}}
\newcommand{\NgridCompl}{\ensuremath{\eps^{-d}\log(1/\eps)}}

\newcommand{\vrestricted}{vertex-restricted}
\newcommand{\scunrestricted}{unrestricted}
\newcommand{\unrestricted}{continuous}
\newcommand{\asymmetric}{directed}

\newcommand{\CodeComment}[1]{\textcolor{blue}{\texttt{#1}}}
\newcommand{\tunnelPrice}{\Algorithm{price}\xspace} 
\newcommand{\matching}{matching\xspace}
\newcommand{\cubcZ}{\widehat{\cZ}}

\newcommand{\Eqlab}[1]{\label{equation:#1}}
\newcommand{\gridMin}{\ensuremath{\eps L/4}}